\documentclass[a4paper,UKenglish,cleveref, autoref, thm-restate, numberwithinsect]{lipics-v2021}

\pdfoutput=1 
\hideLIPIcs  

\title{On (In)approximability of MaxMin Independent Set Reconfiguration}

\author{Hung P. Hoang}{Algorithms and Complexity Group, Faculty of Informatics, TU Wien, Austria \and \url{www.inf.ethz.ch/~hoangp} }{phoang@ac.tuwien.ac.at}{https://orcid.org/0000-0001-7883-4134}{Austrian Science Foundation (FWF, projects 10.55776/Y1329 and ESP1136425)}
\author{Naoto Ohsaka}{CyberAgent, Inc., Tokyo, Japan \and \url{https://todo314.github.io/} }{ohsaka\_naoto@cyberagent.co.jp}{https://orcid.org/0000-0001-9584-4764}{}
\author{Rin Saito}{Graduate School of Information Sciences, Tohoku University, Japan \and \url{https://srin728.github.io/} }{rin.saito@dc.tohoku.ac.jp}{https://orcid.org/0000-0002-3953-4339}{JST SPRING Grant Number JPMJSP2114}
\author{Yuma Tamura}{Graduate School of Information Sciences, Tohoku University, Japan \and \url{https://www.ecei.tohoku.ac.jp/alg/tamura/} }{tamura@tohoku.ac.jp}{https://orcid.org/0009-0001-5479-7006}{JSPS KAKENHI Grant Number JP25K21148}

\authorrunning{H. P. Hoang, N. Ohsaka, R. Saito, and Y. Tamura} 

\Copyright{Hung P. Hoang, Naoto Ohsaka, Rin Saito, and Yuma Tamura} 



\relatedversion{} 

\acknowledgements{
We wish to thank Shuichi Hirahara for helpful conversations.
We are also grateful to the anonymous referees for their helpful comments, in particular for pointing out the reference \cite{bhangale2022ug},
which removes the Unique Games Conjecture assumption in \cite{MR2804389}.
This work was initiated at the \emph{5th Combinatorial Reconfiguration Workshop} (CoRe 2024) in Fukuoka, Japan, in October 2024.
We would like to thank the organisers and participants for the inspiring atmosphere.
}

\nolinenumbers 

\usepackage{xcolor}
\usepackage[dvipsnames]{xcolor}
\usepackage{xspace}
\usepackage{tabularx}
\usepackage{booktabs}
\usepackage{environ}
\usepackage{xpatch}
\usepackage{mdframed}
\usepackage{mathtools}
\usepackage{mleftright}
\mleftright

\usepackage{tikz}
\usetikzlibrary{arrows.meta, positioning}

\hypersetup{colorlinks=true,citecolor=Blue,linkcolor=Red,urlcolor=Magenta}

\theoremstyle{definition}
\newtheorem{definition2}[theorem]{Definition}

\numberwithin{equation}{section}

\crefname{equation}{Eq.}{Eqs.}

\newcommand{\reco}{\leftrightsquigarrow}
\newcommand{\isr}{\textsc{ISR}\xspace}
\newcommand{\mmisr}{\textsc{MMISR}\xspace}
\newcommand{\mis}{\textsc{MIS}\xspace}
\newcommand{\mbb}{\textsc{MBB}\xspace}
\newcommand{\defeq}{\coloneq}

\newcommand{\iini}{I_{\mathrm{ini}}}
\newcommand{\itar}{I_{\mathrm{tar}}}
\newcommand{\imin}{\varphi}
\newcommand{\sqI}{\mathscr{S}}

\usepackage[full,sanserif,disableredefinitions]{complexity} 
\newclass{\XL}{XL}
\newcommand{\fpt}{$\mathsf{FPT}$\xspace}

\newcommand{\N}{\mathbb{N}}
\renewcommand{\epsilon}{\varepsilon}

\newcommand{\bigO}{O}

\DeclarePairedDelimiter{\floor}{\lfloor}{\rfloor}
\DeclarePairedDelimiter{\ceil}{\lceil}{\rceil}

\newcommand{\tw}{\operatorname{tw}}
\newcommand{\bw}{\operatorname{bw}}

\DeclareMathOperator{\val}{\mathsf{val}}
\DeclareMathOperator{\opt}{\mathsf{opt}}
\let\poly\relax\DeclareMathOperator*{\poly}{\mathrm{poly}}
\let\polylog\relax\DeclareMathOperator*{\polylog}{\mathrm{polylog}}

\definecolor{defblue}{rgb}{0, 0.4, 0.796}
\newcommand{\defi}[1]{\textcolor{defblue}{\emph{#1}}}

\makeatletter
\newcommand{\problemtitle}[1]{\gdef\@problemtitle{#1}}
\newcommand{\probleminput}[1]{\gdef\@probleminput{#1}}
\newcommand{\problemtask}[1]{\gdef\@problemtask{#1}}
\NewEnviron{problem}{
\problemtitle{}\probleminput{}\problemtask{}
\BODY
\noindent
\begin{tabularx}{\textwidth}{
	 						l X c}
	\multicolumn{2}{
						l}{\@problemtitle} \\
	\textbf{Input:} & \@probleminput \\
	\textbf{Output:} & \@problemtask
\end{tabularx}
}
\makeatother

\makeatletter
\xpatchcmd{\endmdframed}
  {\aftergroup\endmdf@trivlist\color@endgroup}
  {\endmdf@trivlist\color@endgroup\@doendpe}
  {}{}
\makeatother

\surroundwithmdframed[linecolor=black,linewidth=1pt,skipabove=1mm,skipbelow=1mm,leftmargin=0mm,rightmargin=0mm,innerleftmargin=1mm,innerrightmargin=1mm,innertopmargin=1mm,innerbottommargin=1mm]{problem}

\begin{document}

\maketitle

\begin{abstract}
In the \textsc{Independent Set Reconfiguration} problem under the Token Addition/Removal rule, given a graph $G$ and two independent sets $I$ and $J$ of $G$, we want to transform $I$ into $J$ by adding and removing vertices, such that all the sets throughout the process are independent sets.
Its approximate version called \textsc{MaxMin Independent Set Reconfiguration} aims to maximise the minimum size of the independent sets in the process above.
We study the (in)approximability of this problem for general graphs as well as restricted graph classes.
Firstly, on general graphs, we obtain a polynomial-time $(n / \log n)$-factor approximation algorithm, complementing the $\mathsf{PSPACE}$-hardness of $n^{\Omega(1)}$-factor approximation due to Hirahara and Ohsaka [STOC 2024, ICALP 2024] and the $\mathsf{NP}$-hardness of $n^{1-\varepsilon}$-factor approximation due to Ito, Demaine, Harvey, Papadimitriou, Sideri, Uehara, and Uno [TCS 2011].
Secondly, we present a polynomial-time approximation algorithm for degenerate graphs as well as $\mathsf{FPT}$-approximation schemes for bounded-treewidth graphs and $H$-minor-free graphs.
Lastly, we extend the above inapproximability results to bounded-degree graphs, graphs of bandwidth $n^{\frac{1}{2}+\Theta(1)}$, and bipartite graphs.
\end{abstract}

\clearpage
{\tableofcontents}

\clearpage
\setcounter{page}{1}

\section{Introduction}
\label{sec:intro}
Many combinatorial problems require solutions to be updated over time rather than recomputed from scratch. 
In such scenarios, intermediate solutions have to stay feasible, as abrupt changes are considered undesirable or impossible. 
This perspective is formalised by the framework of \emph{combinatorial reconfiguration}, which studies step-by-step transformations between feasible solutions of a combinatorial problem, which is referred to as the \emph{source problem}.
Each transformation must follow prescribed \emph{reconfiguration rules} and feasibility must be preserved throughout the entire sequence.

\textsc{Independent Set Reconfiguration} (\isr) \cite{hearn2005pspace,hearn2009games,kaminski2012complexity}
is a well-studied reconfiguration problem.
The source problem of \isr is \textsc{Independent Set};
i.e., the feasible solutions are the independent sets of a graph $G$.\footnote{
An \emph{independent set} is defined as a subset of vertices such that no two vertices are adjacent.
}
There are three popular reconfiguration rules, where we view an independent set as tokens placed on the vertices of $G$, described as follows:
(1)~Under the \emph{Token Jumping} rule \cite{kaminski2012complexity}, we may move one token from any vertex to any other vertex.
(2)~Under the \emph{Token Sliding} rule \cite{hearn2005pspace,hearn2009games}, we can only move a token along an edge of $G$.
(3)~Under the \emph{Token Addition/Removal} rule \cite{MR2797748}, we may add or remove a single token at each step, provided that the resulting independent set has size at least a given threshold.
Note that we can only apply a reconfiguration rule to transform an independent set to another independent set; that is, we can only transform between feasible solutions.

The common task is to decide if there exists a sequence to transform between two given independent sets $\iini$ and $\itar$ using one of the rules above.
Such a sequence of independent sets is called a \emph{reconfiguration sequence}.
For each of the three rules, \isr is $\PSPACE$-complete~\cite{hearn2005pspace,hearn2009games,MR2797748,kaminski2012complexity}, even for planar graphs of bounded bandwidth~\cite{MR3452428, wrochna2018reconfiguration}.
On bipartite graphs, \isr is $\PSPACE$-complete for the Token Sliding rule, while it is $\NP$-complete for the other two rules~\cite{lokshtanov2019complexity}.

In order to circumvent this hardness, a common method is to analyse the problem under the framework of parameterised complexity~\cite{CyganFKLMPPS15book}.
This paradigm aims to confine the intractability of a problem to certain well-defined \emph{parameters} of the input, allowing efficient algorithms when these parameters are small.
For an overview on the parameterised complexity of \isr, see \cref{subsec:related} and the recent survey~\cite{MR4782413}.

Another approach involves approximation algorithms, which trade off optimality for improved running time.
We note that, although approximability of reconfiguration problems often refers to that of the \emph{shortest reconfiguration sequence},
e.g., \cite{miltzow2016approximation,yamanaka2015swapping,heath2003sorting},
approximating the shortest reconfiguration sequence length for \isr is at least as hard as solving \isr itself.\footnote{
The shortest length for a YES-instance of \isr is at most $2^n$,
while for a NO-instance, the shortest length can be thought of as infinite since no feasible reconfiguration sequence exists.
Therefore, any approximation algorithm for the shortest reconfiguration sequence can distinguish between YES and NO instances of \isr.
}
In this paper, we study an approximate version of \isr, the so-called \textsc{MaxMin Independent Set Reconfiguration} (\mmisr) \cite{MR2797748}, defined as follows:
Given a pair of independent sets $\iini$ and $\itar$ of a graph~$G$, we want to transform $\iini$ into $\itar$ with the goal of maximising the minimum size of the independent sets in this sequence. 
This problem is only meaningful under the Token Addition/Removal rule, and further, to make sure it is always a YES-instance, the minimum size threshold is set to zero.
See \cref{sec:prelims} for a formal definition of the problem.

=We review known results on the approximability of \mmisr.
It is $\NP$-hard to approximate \mmisr on $n$-vertex graphs within an $n^{1-\varepsilon}$-factor for any $\varepsilon > 0$~\cite{MR2797748,zuckerman2007linear}.
Ohsaka~\cite{ohsaka2023gap} showed the $\PSPACE$-hardness of constant-factor approximation for graphs of maximum degree $3$, and
Hirahara and Ohsaka~\cite{hirahara2024optimal,MR4764919} showed
the $\PSPACE$-hardness of $n^{\Omega(1)}$-factor approximation for general graphs on $n$ vertices and
the $\PSPACE$-hardness of $\Delta^{\Omega(1)}$-factor approximation for graphs of maximum degree $\Delta$. 
Despite these strong hardness of approximation results, to the best of our knowledge,
no approximation algorithm for \mmisr is currently known.
See \cref{subsec:related} for existing results on approximability of \mmisr.

\subsection{Contributions}
In this paper, we investigate both approximation algorithms and inapproximability for \mmisr.
See \cref{tab:algo,tab:hardness} for summaries of the algorithmic and hardness results, respectively.
Note that all our algorithms also output a reconfiguration sequence rather than just the size of the smallest set.

\begin{table}[ht!]
    \centering
    \caption{Summary of our algorithmic results. Here, $n$ is the number of vertices in the graph, $\imin = \min\{|\iini|, |\itar|\}$, $\epsilon$ is an arbitrary positive real, and $f$ is some computable function. Note that the first algorithm for bounded-treewidth graphs runs in polynomial time, provided that a tree decomposition of width $k$ is given.}
    \label{tab:algo}
    \begin{tabular}{c c c c }
        \toprule
        Class of graphs & Factor & Run time &  Reference \\
        \midrule
        General & $n / \log n$ & $n^{\bigO(1)}$ & \cref{thm:general_aprox} \\
        Degeneracy $\leq k$ & 
        $k$
        & $n^{\bigO(1)}$ & \cref{thm:algo:degeneracy} \\
        Treewidth $\leq k$ & $\frac{\imin}{\imin-\bigO(k \log \imin)}$ & $n^{\bigO(1)}$ & \cref{thm:algo:treewidth} \\
        Treewidth $\leq k$ & $1 + \varepsilon$ & $f(k, \varepsilon)\cdot n^{\bigO(1)}$ & \cref{thm:treewidth_PTAS} \\
        $H$-minor-free & $1 + \varepsilon$ & $f(|V(H)|, \varepsilon)\cdot n^{\bigO(1)}$ & \cref{thm:minor_free} \\
        \bottomrule
    \end{tabular}
\end{table}

\subparagraph{Approximation algorithms.}
Our first contribution is an $(n/\log n)$-factor polynomial-time approximation algorithm for \mmisr on general $n$-vertex graphs (\cref{thm:general_aprox}).
This is the first non-trivial approximation algorithm for \mmisr, and
complements the $\NP$-hardness of $n^{1-\varepsilon}$-factor approximation \cite{MR2797748,zuckerman2007linear}.

\begin{figure}
    \centering
    \begin{tikzpicture}[
        box/.style={
            draw,
            rounded corners,
            minimum width=3cm,
            minimum height=0.75cm,
            align=center
        },
        >=Stealth
    ]
    
    \node[box] (deg) {Degeneracy};
    
    \node[box, below=2em of deg, xshift=-2.5cm] (tw) {Treewidth};
    \node[box, below=2em of deg, xshift= 2.5cm] (md) {Maximum degree};
    
    \node[box, below=2em of tw, xshift=2.5cm] (bw) {Bandwidth};
    
    \draw[->] (bw) -- (tw);
    \draw[->] (bw) -- (md);
    
    \draw[->] (tw) -- (deg);
    \draw[->] (md) -- (deg);
    
    \end{tikzpicture}
    \caption{Relationships between some parameters in this paper. An arrow from a parameter $A$ to a parameter $B$ indicates that if $A$ is bounded, then $B$ is also bounded.
    See \cref{sec:prelims} for the definitions of these graph classes.}
    \label{fig:parameters}
\end{figure}
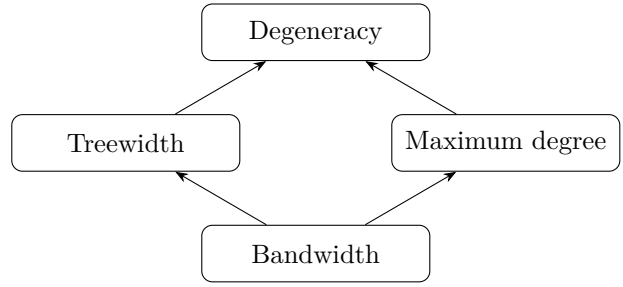

Our next contribution is to develop approximation algorithms for \mmisr on special graph classes,
by leveraging techniques from parameterised complexity.
See \cref{fig:parameters} for the relationships between graph parameters examined in this paper and \cref{sec:prelims} for the formal definitions of these parameters.
First, we show a polynomial-time algorithm for \emph{$k$-degenerate graphs}
whose approximation factor is 
$\approx k$
(\cref{thm:algo:degeneracy}).

Second, for \emph{graphs of treewidth at most $k$},
we show a polynomial-time approximation algorithm for \mmisr whose approximation factor is $\approx \frac{\imin}{\imin - \bigO(k\log\imin)}$, where $\imin$ is the input independent set size (\cref{thm:algo:treewidth}),
assuming that a tree decomposition of width $k$ is given as input.
As an application of this algorithm, we show an \emph{\fpt-approximation scheme} (FPT-AS) parameterised by treewidth $k$ (\cref{thm:treewidth_PTAS}); i.e.,
for every $\epsilon > 0$, there exists a $(1+\varepsilon)$-factor approximation algorithm that runs in time $f(k, \epsilon) \cdot n^{\bigO(1)}$ for some computable function $f$.
Our FPT-AS is obtained by combining \cref{thm:algo:treewidth} with
an \fpt algorithm for \mmisr, which is based on \cite{AgrawalHM25,kaminski2012complexity}.
Note that an FPT-AS implies an \emph{efficient polynomial-time approximation scheme} (EPTAS) for each fixed $k$.\footnote{
This is stronger than \cref{thm:algo:degeneracy},
which implies an $\bigO(k)$-factor approximation algorithm since the degeneracy is bounded by the treewidth.
}

Third, we show an FPT-AS for \emph{$H$-minor-free graphs} parameterised by the size of $H$ (\cref{thm:minor_free}).
$H$-minor-free graphs form a subclass of bounded-degeneracy graphs~\cite{Thomason2001} and include planar graphs and bounded-genus graphs.
Specifically, for every $\epsilon > 0$ and for graphs that exclude a fixed graph $H$ as a minor,
there exists a $(1+\varepsilon)$-factor approximation algorithm that runs in time $f(|V(H)|, \epsilon) \cdot n^{\bigO(1)}$ for some computable function $f$.

\subparagraph{Inapproximability results.}
Our last contribution is that we establish inapproximability results for \mmisr on restricted graph classes.
Specifically, we first prove the $\NP$-hardness of $\Theta(\sqrt{\Delta})$-factor approximation for graphs of maximum degree $\Delta$ (\cref{thm:hardness:degree}).
This hardness result is quantitatively stronger than
the $\PSPACE$-hardness of $\Delta^{\epsilon_0}$-factor approximation \cite{hirahara2024optimal},
where $\epsilon_0$ may be extremely small (e.g., $\epsilon_0 = 0.001$).
Then, we prove the $\NP$-hardness of $\Theta\bigl(n^{4\delta^2-\epsilon}\bigr)$-factor approximation
for $n$-vertex graphs of bandwidth $\bigO\bigl(n^{\frac{1}{2}+\delta}\bigr)$, where
$\delta \in \left(0,\frac{1}{2}\right)$ is a positive real and $\epsilon$ is any small positive real
(\cref{thm:hardness:bandwidth}).
This result means that for graphs of bandwidth, treewidth, and degeneracy $n^{\frac{1}{2}+\Theta(1)}$,
\mmisr cannot be approximated within a polynomial factor in polynomial time unless $\cP = \NP$.
Finally, we prove that
\mmisr on $n$-vertex bipartite graphs cannot be approximated within a factor of $n^{1-\epsilon}$ for any $\epsilon > 0$,
assuming the Small Set Expansion Hypothesis (SSEH)~\cite{raghavendra2010graph} and $\NP \nsubseteq \BPP$ (\cref{thm:hardness:bipartite}).
Therefore, \mmisr on bipartite graphs is as hard as on general graphs with respect to the approximation factor.

\begin{table}[ht!]
    \centering
    \caption{Summary of inapproximability results for \mmisr,
    where $n$ is the number of vertices in the graph,
    $\delta \in \left(0,\frac{1}{2}\right)$ is any positive real, and
    $\epsilon$ is any small positive real. These hardness results exclude polynomial-time approximation algorithms for \mmisr on each graph class under a particular hypothesis.}
    \label{tab:hardness}
    \begin{tabular}{c c c c }
        \toprule
        Class of graphs & Factor & Assumptions &  Reference \\
        \midrule
        General & $n^{1-\varepsilon}$ & $\NP \neq \cP$ & \cite{MR2797748,zuckerman2007linear} \\
        General & $n^{\Omega(1)}$ & $\PSPACE \neq \cP$ & \cite{MR4764919} \\
        Maximum degree $\Delta$ & $\Delta^{\Omega(1)}$ & $\PSPACE \neq \cP$ & \cite{hirahara2024optimal} \\
        Maximum degree $\Delta$ & $\Theta\bigl(\sqrt{\Delta}\bigr)$ & $\NP \neq \cP$ & \cref{thm:hardness:degree} \\
        Bandwidth $\bigO\bigl(n^{\frac{1}{2} + \delta}\bigr)$ & $\Theta\bigl(n^{4\delta^2-\epsilon}\bigr)$ & $\NP \neq \cP$ & \cref{thm:hardness:bandwidth} \\
        Bipartite & $n^{1-\varepsilon}$ & SSEH, $\NP \nsubseteq \BPP$ & \cref{thm:hardness:bipartite} \\
        \bottomrule
    \end{tabular}
\end{table}

\subparagraph{Discussions.}
A natural question is whether our $(n/\log n)$-approximation can be improved by applying stronger static approximation algorithms for Maximum Independence Set (MIS) on general graphs (see an overview in \cref{subsec:related} below).
We do not know how to do this in a black-box manner.
The main obstacle is that \mmisr requires not just one large independent set but a reconfiguration sequence where every independent set is large throughout.
By contrast, an approximation algorithm for MIS only returns a single large independent set, and it is not clear how to connect it with the prescribed initial and target sets such that all the independent sets throughout the reconfiguration process remain large.
Therefore, improving the approximation factor for \mmisr seems to require techniques that are aware of the reconfiguration rules.

Next, we compare our algorithmic and hardness results.
Since the degeneracy is bounded by the maximum degree,
\cref{thm:algo:degeneracy,thm:hardness:degree} imply that
the optimal approximation factor for \mmisr on graphs of degeneracy $k$ lies between $\Theta(\sqrt{k})$ and $\Theta(k)$.
Since the treewidth is bounded by the bandwidth,
\cref{thm:treewidth_PTAS,thm:hardness:bandwidth} imply that
\mmisr on graphs of treewidth $k$ admits an EPTAS if $k=\bigO(1)$, while
does not have a polynomial-factor approximation algorithm if $k=n^{\frac{1}{2}+\Theta(1)}$.
An immediate open question is thus
if there exists an (E)PTAS for \mmisr on graphs of superconstant treewidth $k$ (e.g., $k = \log n$).

\subsection{Organisation}
The rest of this paper is organised as follows.
In \cref{subsec:related}, we review related work.
In \cref{sec:prelims}, we formally define the \mmisr problem and introduce related notions.
In \cref{sec:tech_overview}, we present technical overviews of our results.
In \cref{sec:algo:general}, we develop an $\frac{n}{\log n}$-factor approximation algorithm for \mmisr (\cref{thm:general_aprox}).
In \cref{sec:algo}, we construct approximation algorithms for \mmisr on special graph classes
(\cref{thm:algo:degeneracy,thm:algo:treewidth,thm:treewidth_PTAS,thm:minor_free}).
In \cref{sec:hardness}, we show inapproximability results for \mmisr on restricted graph classes
(\cref{thm:hardness:bipartite,thm:hardness:degree,thm:hardness:bandwidth}).

\subsection{Related work}
\label{subsec:related}

\subparagraph{Parameterised complexity of \textsc{Independent Set Reconfiguration}.}
The parameterised complexity of \textsc{Independent Set Reconfiguration} under the Token Jumping rule (\textsc{ISR-TJ}) is well-studied with respect to the independent set size~$k$.
On general graphs, \textsc{ISR-TJ} is known to be \W[1]-hard~\cite{ISR:ItoKOSUY20,ISR:MouawadN0SS17} and \XL-complete~\cite{ISR:BodlaenderGS21}.
As with the classical \textsc{Independent Set} problem, FPT algorithms for \textsc{ISR-TJ} have been developed for several sparse graph classes.
Specifically, Bousquet, Mary and Parreau~\cite{BousquetMP17} showed that the problem is fixed-parameter tractable on biclique-free graphs, a wide class that encompasses planar graphs, bounded-degeneracy graphs, and nowhere dense graphs.
For a comprehensive review of these and related results, we refer the reader to the recent survey by Bousquet, Mouawad, Nishimura and Siebertz~\cite{MR4782413}.

\subparagraph{Approximability of \mmisr and other reconfiguration problems.}
For a reconfiguration problem,
its \emph{approximate version} allows to use infeasible solutions,
but requires optimising the ``worst'' feasibility along the reconfiguration sequence.
Ito, Demaine, Harvey, Papadimitriou, Sideri, Uehara, and Uno \cite{MR2797748}
showed that approximate versions of
\textsc{Independent Set Reconfiguration} and \textsc{SAT Reconfiguration} are $\NP$-hard to approximate.
Note that the reduction in \cite{MR2797748}
along with the inapproximability results of \textsc{Maximum Independent Set} due to
H\r{a}stad~\cite{MR1687331} and Zuckerman~\cite{zuckerman2007linear}
implies the $\NP$-hardness of $n^{1-\epsilon}$-approximation for any $\epsilon > 0$.
$\PSPACE$-hardness of approximation for reconfiguration problems was posed as an open problem by \cite{MR2797748}.
Ohsaka~\cite{ohsaka2023gap} postulated a reconfiguration analogue of
the PCP theorem \cite{arora1998probabilistic,arora1998proof},
called the \emph{Reconfiguration Inapproximability Hypothesis} (RIH),
and proved that assuming RIH, approximate versions of several reconfiguration problems are $\PSPACE$-hard to approximate, including \mmisr.
Recently,
Hirahara and Ohsaka~\cite{MR4764919} and
Guruswami, Karthik C.~S., Manurangsi, Ren, and Wu~\cite{guruswami2025inapproximability}
independently gave a proof of RIH,
thereby resolving the open problem of \cite{MR2797748} affirmatively.
Ohsaka~\cite{ohsaka2025yet} also gave an alternative proof of RIH.

Since the resolution of RIH,
the following results have been obtained on the hardness of approximating \mmisr:
Ohsaka~\cite{ohsaka2023gap} showed that
\mmisr on graphs of maximum degree $3$ is $\PSPACE$-hard to approximate within a constant factor.
Hirahara and Ohsaka showed that
\mmisr on $n$-vertex graphs is $\PSPACE$-hard to approximate
within a factor of $n^{\Omega(1)}$ \cite{MR4764919}, and
\mmisr on graphs of maximum degree $\Delta$ is $\PSPACE$-hard to approximate
within a factor of $\Delta^{\Omega(1)}$ \cite[Theorem~3.4]{hirahara2024optimal}.
In this paper, we 
improve the known inapproximability factors for \mmisr quantitatively.

Approximation algorithms and inapproximability results have also been studied for other reconfiguration problems,
including
\textsc{2-CSP Reconfiguration}~\cite{guruswami2025inapproximability,ohsaka2024alphabet,ohsaka2024gap,ohsaka2025approximate},
\textsc{Set Cover Reconfiguration}~\cite{hirahara2024optimal,guruswami2025inapproximability},
\textsc{$k$-Colouring Reconfiguration}~\cite{hirahara2025asymptotically},
\textsc{$k$-SAT Reconfiguration}~\cite{hirahara2025asymptoticallya},
\textsc{Subset Sum Reconfiguration}~\cite{ito2014approximability}, and
\textsc{Submodular Reconfiguration}~\cite{ohsaka2022reconfiguration}.

\subparagraph{Approximability of \textsc{Maximum Independent Set}.}
For the classical \textsc{Maximum Independent Set} (\mis) problem, it is $\NP$-hard to approximate within a factor of $n^{1-\varepsilon}$ for any $\epsilon > 0$ \cite{MR1687331,zuckerman2007linear}, while the best algorithm by Feige~\cite{feige2004approximating} achieves the approximation factor of $\bigO\left(\frac{n(\log \log n)^2}{(\log n)^3}\right)$.
For graphs with maximum degree $\Delta$, it is hard to approximate within a factor of $\bigO\left(\frac{\Delta}{(\log \Delta)^2}\right)$ \cite{bhangale2022ug}, while there exists $\bigO\left(\frac{\Delta \log \log \Delta}{\log \Delta}\right)$-approximation algorithm~\cite{MR1936662}.
On planar graphs, Baker~\cite{Baker94} showed a PTAS, and this has later been generalised into a PTAS for $H$-minor-free graphs~\cite{DemaineHK05bidimension}.

\subparagraph{Parameterised approximation schemes.}
\fpt-approximation schemes are a common approach in parameterised approximation, which combines the techniques from parameterised complexity and approximation in order to overcome the hardness in both paradigms.
See the surveys by Marx~\cite{marx2008parameterised} and by Feldmann, Karthik C.~S., Lee, and Manurangsi~\cite{feldmann2020survey}.

\section{Preliminaries}
\label{sec:prelims}
For a positive integer $n \in \mathbb{N}$, we write $[n] \coloneq \{1,2,\ldots,n\}$.
For a graph $G$, we denote its vertex and edge sets by $V(G)$ and $E(G)$, respectively.
For a subset $X \subseteq V(G)$, we denote by $G[X]$ the induced subgraph of $G$ on $X$.
The \defi{independence number} of $G$, denoted by $\alpha(G)$, is the size of a maximum independent set of $G$.
In this paper, unless otherwise stated, we assume that $\alpha(G) > 1$ (i.e., $G$ is not a complete graph).

For two independent sets $\iini$ and $\itar$ of a graph $G$, a \defi{reconfiguration sequence} from $\iini$ to $\itar$ is a sequence $(\iini = I^{(0)}, I^{(1)}, \dots, I^{(t)} = \itar)$ such that for $i \in [t]$, $I^{(i)}$ is an independent set of $G$, and either $I^{(i-1)} \subseteq I^{(i)}$ or $I^{(i-1)} \supseteq I^{(i)}$.\footnote{The Token Addition/Removal rule in \cref{sec:intro} only allows adding or removing a single vertex at each step. However, as the length of the reconfiguration sequence plays no role in our optimisation problem, we use a slightly relaxed but equivalent formulation of this rule.}
We also call it an \defi{$(\iini,\itar)$-reconfiguration sequence}.
Unless otherwise stated, we assume that $\iini$ and $\itar$ are distinct.
For a reconfiguration sequence
$\mathscr{S}$ of $G$,
we define its \defi{value} as
\begin{align*}
    \val_G(\mathscr{S}) \coloneq \min_{I \in \mathscr{S}} |I|.
\end{align*}

\begin{problem}
      \problemtitle{\textsc{MaxMin Independent Set Reconfiguration} (\mmisr)}
      \probleminput{A graph $G$ and two independent sets $\iini$ and $\itar$}
      \problemtask{The maximum value of $\val_G(\mathscr{S})$ among all $(\iini,\itar)$-reconfiguration sequences}
\end{problem}

We denote by $\opt_G(\iini \reco \itar)$ the optimal value of \mmisr.

\subparagraph{Parameterised complexity.}
In parameterised complexity~\cite{DowneyFellows13,CyganFKLMPPS15book}, the complexity of a problem is studied not only with respect to the input size, but also with respect to some problem parameter(s). 
The core idea behind parameterised complexity is that the combinatorial explosion resulting from the $\NP$-hardness of a problem can sometimes be confined to certain structural parameters that are small in practical settings. 

Formally, a \defi{parameterised problem} $Q$ is a subset of $\Omega^* \times \N$, where $\Omega$ is a fixed alphabet. 
Each instance of $Q$ is a pair $(x, \kappa)$, where $\kappa \in \N$ is called the \defi{parameter}. 
A parameterised problem $Q$ is \defi{fixed-parameter tractable} (\fpt) if there is an algorithm, called a \defi{fixed-parameter algorithm},  that decides whether an input $(x, \kappa)$ is a member of $Q$ in time $f(\kappa) \cdot |x|^{\bigO(1)}$, where $f$ is a computable function and $|x|$ is the input instance size.  
The class \fpt denotes the class of all fixed-parameter tractable parameterised problems.

Extending this notion to optimisation, an \defi{\fpt-approximation scheme} (FPT-AS) for a parameterised problem is an algorithm that, given an instance $(x, k)$ and an error parameter $\varepsilon > 0$, computes a $(1+\varepsilon)$-approximate solution in $f(k, \varepsilon) \cdot |x|^{O(1)}$ time.

\subparagraph{Bandwidth.}
For a graph $G$, its \defi{bandwidth} $\bw(G)$ is defined as $\bw(G) = \min_{\sigma} \max_{uv \in E(G)} |\sigma(u) - \sigma(v)|$, where the minimum is taken over all bijections $\sigma \colon V(G) \to [|V(G)|]$.

\subparagraph{Degeneracy.}
A graph $G$ is \defi{$d$-degenerate} if every subgraph of $G$ contains a vertex of degree at most $d$.
The \defi{degeneracy} of $G$ is the smallest value of $d$ such that $G$ is $d$-degenerate.

\subparagraph{Treewidth.}
A \defi{tree decomposition} of a graph $G = (V, E)$ is a pair $(T, \{X_i\}_{i \in V_T})$, where $T = (V_T, E_T)$ is a tree and each $X_i$ (called a \defi{bag}) is a subset of $V$, satisfying the following three conditions:
\begin{enumerate}
    \item $\bigcup_{i \in V_T} X_i = V$; 
    \item for each edge $uv \in E$, there exists at least one bag $X_i$ that contains both $u$ and $v$; and
    \item for each vertex $v \in V$, the set $\{ i \in V_T \mid v \in X_i \}$ induces a connected subtree of $T$.
\end{enumerate}
The \defi{width} of a tree decomposition $(T, \{X_i\}_{i \in V_T})$ is defined as $\max_{i \in V_T} |X_i| -1$. 
The \defi{treewidth} of a graph $G$, denoted by $\tw(G)$, is the minimum width among all tree decompositions of $G$.

\subparagraph{$H$-minor-free graphs.}
For an edge $e = uv$ of $G$, let $G/e$ denote the graph obtained by contracting $e$, i.e., identifying $u$ and $v$ and subsequently removing all loops and parallel edges.
A graph $H$ is a \defi{minor} of $G$ if $H$ can be obtained from $G$ by a sequence of vertex deletions, edge deletions, and edge contractions.
We say that $G$ is \defi{$H$-minor-free} if it does not contain $H$ as a minor.

\subparagraph{ISR under Token Jumping rule.}
Although this problem is not the focus of the paper, since it is relevant for a result in \cref{sec:algo:tw:fpt}, we also define it here.
For two independent sets $\iini$ and $\itar$ of a graph $G$, a \defi{TJ-reconfiguration sequence} from $\iini$ to $\itar$ is a sequence $(\iini = I^{(0)}, I^{(1)}, \dots, I^{(t)} = \itar)$ such that for $i \in [t]$, $I^{(i)}$ is an independent set of $G$, and $|I^{(i-1)} \setminus I^{(i)}| = |I^{(i)} \setminus I^{(i-1)}| = 1$.

\begin{problem}
      \problemtitle{\textsc{Independent Set Reconfiguration - Token Jumping} (ISR-TJ)}
      \probleminput{A graph $G$ and two independent sets $\iini$ and $\itar$}
      \problemtask{Whether there exists a TJ-reconfiguration sequence from $\iini$ to $\itar$}
\end{problem}

\section{Technical overview}
\label{sec:tech_overview}

We provide an overview of the approximation algorithms as well as the proofs of the inapproximability results.
Throughout this subsection, let $G$ be a graph, and $\iini, \itar$ be two independent sets of $G$.

\subsection{Approximation algorithms}
\subparagraph{General graphs.}

To obtain an $(n / \log n)$-factor approximation algorithm, we introduce the notion of a \emph{$\gamma$-sequence} for a nonnegative integer~$\gamma$.
Given an instance $(G, \iini, \itar)$ of \mmisr, a $\gamma$-sequence $\mathscr{S} = (I^{(1)}, I^{(2)}, \ldots, I^{(t)})$ is a sequence of independent sets of size at least~$\gamma$ such that the following procedure yields a reconfiguration sequence from $\iini$ to $\itar$.
(1) Remove vertices from $\iini$ to obtain $I^{(1)}$.
(2) For each $i = 1, 2, \ldots, t-1$, first add the vertices in $I^{(i+1)} \setminus I^{(i)}$ to obtain $I^{(i)} \cup I^{(i+1)}$, and then remove the vertices in $I^{(i)} \setminus I^{(i+1)}$ to obtain $I^{(i+1)}$.
(3) Add vertices to $I^{(t)}$ to obtain $\itar$.
We therefore require that $I^{(i)} \cup I^{(i+1)}$ is an independent set of $G$ for every $i$, so that all intermediate sets in the above procedure are feasible.
The formal definition is given in \Cref{sec:algo:general}.
By definition, the existence of a $\gamma$-sequence implies a reconfiguration sequence from $\iini$ to $\itar$ of value at least~$\gamma$.

Now consider an arbitrary partition of the vertex set $V(G)$ into disjoint subsets $V_1, V_2, \dots, V_{\ell}$ of equal size for some $\ell$ that we can choose later.
Then we can show the existence of a $\gamma$-sequence $\mathscr{S}$ such that each element of $\mathscr{S}$ is a subset of some (not necessarily the same) $V_j$, and $\mathscr{S}$ yields an $\ell$-factor approximation.
Specifically, the idea of this existence proof is that for each set in an optimal sequence, we select an index $j$ such that the size of its intersection with $V_j$ is maximised.
By concatenating all such intersections, we obtain the desired $\gamma$-sequence.

Next, our approach to find a $\gamma$-sequence as described above is to create an auxiliary graph $H_{\gamma}$ for each value $\gamma \in \{0, \dots, n/\ell\}$.
(Note that $n/\ell$ is the size of every $V_i$.)
The vertices of $H_{\gamma}$ are all the independent sets $I \subseteq V_i$ of $G$ with $|I| \ge \gamma$ for some $i \in \{1,\dots, \ell\}$.
Two such sets are connected by an edge in $H_{\gamma}$ if and only if their union is an independent set of $G$.
A desired $\gamma$-sequence then corresponds to a path in $H_\gamma$ from a subset of $\iini$ to a subset of $\itar$, and we choose the maximum $\gamma \in \{0, \dots, n/\ell\}$ for which such a path exists.

In order to find this path in polynomial time, we want $H_{\gamma}$ to have polynomial size.
Since each $V_i$ has $2^{n/\ell}$ subsets, we choose $\ell = n/\log n$.
We then obtain the approximation factor of $n / \log n$ as discussed above.

\subparagraph{Degenerate graphs.}
Suppose $G$ is $d$-degenerate; that is, every induced subgraph of $G$ contains a vertex of degree at most $d$.
The main idea of our algorithm is as follows.
We maintain two independent sets $A$ and $B$, where initially $A \coloneq \iini$ and $B \coloneq \itar$.
We reconfigure $A$ or $B$ in order to increase the size of $A \cap B$.
In particular, let $X \coloneq A \setminus B$ and  $Y \coloneq B \setminus A$.
Since $G$ is $d$-degenerate, in the subgraph of $G$ induced on $X \cup Y$, there is a vertex $v$ that has degree at most $d$.
If $v \in X$, we remove the neighbours of $v$ in $B$ from $B$ and add $v$ to $B$.
Note that the neighbours of $v$ in $B$ have to be in $Y$, since the vertices in $B \setminus Y$ are also in $A$ and hence are not adjacent to $v$.
In other words, we remove at most $d$ vertices from $B$ and add exactly one vertex to $B$.
If $v \in Y$, we analogously remove at most $d$ neighbours of $v$ in $A$ from $A$ and add $v$ to $A$.
Note that after every iteration, we remove at least one vertex from $X \cup Y$.
Hence, eventually, we have $A = B$.
We can then combine the reconfiguration sequence for $A$ and the reversed reconfiguration sequence for $B$ to obtain a reconfiguration sequence from $\iini$ to $ \itar$.

The algorithm clearly runs in polynomial time.
To analyse the approximation factor, observe that at every step in the reconfiguration sequence for $A$, we remove at most $d$ vertices from $A$ and add exactly one vertex to $A$.
Consequently, after $i$ steps, the size of $A$ decreases by at most $(d-1)i$.
On the other hand, exactly $i$ vertices have been added to $A$ during these $i$ steps. 
Hence, after $i$ steps, the size of the independent set $A$ is at least $\max \left\{ |\iini| - (d-1)i, i \right\} \geq \frac{1}{d}|\iini|$.
This implies all independent sets in the reconfiguration sequence for $A$ have size at least $\frac{1}{d}|\iini|$.
Similarly, all independent sets in the reconfiguration sequence for $B$ have size at least $\frac{1}{d}|\itar|$.
Since the optimal value is at most $\min\{|\iini|,|\itar|\}$, the algorithm achieves an approximation factor of $d$.

\subparagraph{Bounded-treewidth graphs.}
Suppose $G$ has treewidth $k$.
We assume that a tree decomposition of width $k$ is given; otherwise, this can be computed in \fpt time~\cite{Bodlaender96}.
Further, let $\imin = \min\{|\iini|,|\itar| \}$.
We now discuss the main idea of the $\frac{\imin}{\imin - \bigO(k \log \imin)}$-approximation algorithm.
This approach relies on the fact that there exists a \emph{$\frac{2}{3}$-balanced $\iini$-separator} $S$ of size at most $k+1$~\cite{CyganFKLMPPS15book}, i.e., we can partition $V(G) \setminus S$ into two sets $X$ and $Y$ such that there are no edges between $X$ and $Y$, $|X \cap \iini| \leq \frac{2}{3} |\iini|$, and $|Y \cap \iini| \leq \frac{2}{3} |\iini|$.
We then remove the vertices in $\iini \cap S$, recurse on $X$ and $Y$ in some suitable order, and finally add $\itar \cap S$.
Since $|X \cap \iini|$ and $|Y \cap \iini|$ are at most $\frac{2}{3}|\iini|$, the recursion depth is $\bigO(\log \imin)$, and hence the overall runtime is polynomial.
We then carefully analyse the size of the smallest independent set obtained from the recursion, using the order of processing $X$ and $Y$ as well as the bounds on $|S|$, $|X \cap \iini|$, and $|Y \cap \iini|$.

For the \fpt-approximation scheme, we aim to obtain a $(1+\epsilon)$-factor approximation in time $f(k, \epsilon)\cdot n^{\bigO(1)}$ for any $\epsilon > 0$ and some computable function $f$.
If the $\frac{\imin}{\imin - \bigO(k \log \imin)}$-approximation algorithm above already achieves this guarantee, we use it.
Otherwise, $\imin$ must be bounded from above by a function of $k$ and $\epsilon$.
We now apply an \fpt algorithm to solve \mmisr exactly parameterised by $\imin$ and $k$ (\cref{thm:FPT_degeneracy}), which is also an \fpt algorithm parameterised by $k$ and $\epsilon$ in this case.
In order to obtain this \fpt algorithm, we first extend the equivalence between the Token Jumping rule and Token Addition/Removal rule of ISR~\cite{kaminski2012complexity} to an equivalence between ISR under the Token Jumping rule (ISR-TJ) and \mmisr.
We then derive the \fpt algorithm above from a recent \fpt algorithm for ISR-TJ parameterised by degeneracy and the size of the input independent sets~\cite{AgrawalHM25}. 

\subparagraph{$H$-minor-free graphs.}
Let $\eta \coloneq \max\{|\iini|,|\itar|\}$.
We employ a generalisation of Baker's technique, which is a framework for developing polynomial-time approximation schemes (PTASs) for various problems on planar graphs~\cite{Baker94}.
It is known that the vertex set of any planar graph $G$ can be partitioned into $V_1, \dots, V_{k+1}$ such that, for each $i \in [k+1]$, the subgraph $G_i$ induced by $V(G) \setminus V_i$ has bounded treewidth.
Moreover, by the pigeonhole principle, there is an integer $j \in [k+1]$ such that $\max\{|\iini \cap V_j|, |\itar \cap V_j| \} \le \frac{2\eta}{k+1}$, which implies that $\imin \coloneq \min\{|\iini\cap V(G_j)|, |\itar\cap V(G_j)| \} \geq \left(1-\frac{2}{k+1}\right)\eta$.

The high-level idea of our algorithm is as follows:
first remove the vertices in $V_j \cap \iini$ from $\iini$, then apply the \fpt-approximation scheme for bounded-treewidth graphs to $G_j$, and finally add the vertices in $V_j \cap \itar$.
By appropriately setting $k$ and using the bound $\imin \ge \left(1-\frac{2}{k+1}\right)\eta$, this yields an \fpt-approximation scheme for planar graphs.
A similar argument can be established for $H$-minor-free graphs by using the decomposition theorem of Demaine, Hajiaghayi, and Kawarabayashi~\cite{DemaineHK05bidimension}.

\subsection{Inapproximability results}
We first review the $\NP$-hardness of approximating \mmisr on general graphs \cite{MR2797748}.
The proof is based on a gap-preserving reduction from \textsc{Maximum Independent Set} (\mis) to \mmisr.
Let $G$ be an $n$-vertex graph as an instance of \mis.
For a positive integer $k$,
we create a complete balanced bipartite graph $K \defeq K_{k,k}$ with bipartition $(L,R)$, where $|L|=|R|=k$.
Define
$H$ as the disjoint union of $G$ and $K$, $\iini \defeq L$, and $\itar \defeq R$,
which yields an instance $(H,\iini,\itar)$ of \mmisr.
Let $\alpha(G)$ denote the size of maximum independent sets of $G$ and
$\opt_H(\iini \reco \itar)$ denote the optimal value of \mmisr.
We obtain the following relation between $\alpha(G)$ and $\opt_H(\iini \reco \itar)$.

\begin{observation}
\label{obs:intro:hard}
$
    \opt_H(\iini \reco \itar)
    = \min\bigl\{ |L|, |R|, \alpha(G) \bigr\}
    = \min\bigl\{ k, \alpha(G) \bigr\}.
$
\end{observation}
\begin{proof}
Let $I_G$ be any maximum independent set of $G$, where $|I_G| = \alpha(G)$.
Consider a reconfiguration sequence $\sqI$ from $\iini$ to $\itar$, passing through $I_G$,
obtained by the following procedure:
\begin{description}
    \item[Step 1.] add all $\alpha(G)$ vertices of $I_G$ to obtain $L \cup I_G$.
    \item[Step 2.] remove all $k$ vertices of $L$ to obtain $I_G$.
    \item[Step 3.] add all $k$ vertices of $R$ to obtain $R \cup I_G$.
    \item[Step 4.] remove all $\alpha(G)$ vertices of $I_G$ to obtain $R$.
\end{description}
Observe that the objective value of $\sqI$ is at least $\min\bigl\{ |L|, |R|, \alpha(G) \bigr\}$, and
thus so is $\opt_H(\iini \reco \itar)$.

On the other hand,
let $\sqI^*$ be an optimal reconfiguration sequence from $\iini$ to $\itar$.
Since $K$ is a complete bipartite graph with bipartition $(\iini,\itar)=(L,R)$,
there must exist an independent set $I$ in $\sqI^*$ such that
$I \cap (L \cup R) = \emptyset$,
implying that $|I| = |I \cap V(G)| \leq \alpha(G)$.
Therefore, the objective value of $\sqI^*$ is at most $\min\bigl\{ |L|, |R|, \alpha(G) \bigr\}$, and
thus so is $\opt_H(\iini \reco \itar)$.
\end{proof}

Recall that it is $\NP$-hard to distinguish whether
$\alpha(G) \geq n^{1-\epsilon}$ or $\alpha(G) \leq n^{\epsilon}$
for any small $\epsilon > 0$ \cite{MR1687331,zuckerman2007linear}.
By applying \cref{obs:intro:hard}
with $k \defeq n$,
it is $\NP$-hard to distinguish whether
$\opt_H(\iini \reco \itar) \geq n^{1-\epsilon} = \Theta(|V(G)|^{1-\epsilon})$ or
$\opt_H(\iini \reco \itar) \leq n^{\epsilon} = \Theta(|V(G)|^{\epsilon})$
for any small $\epsilon > 0$.
Therefore,
\mmisr on $n$-vertex graphs is $\NP$-hard to approximate within a factor of $n^{1-\epsilon}$,
as desired.

Consider now proving the $\NP$-hardness of approximating \mmisr on restricted graph classes.
To this end, we modify the above gap-preserving reduction from \mis to \mmisr as follows
(see \cref{sec:hardness} for details):
\begin{description}
    \item[(Maximum degree $\Delta$)]
        By \cite{bhangale2022ug},
        it is $\NP$-hard to approximate \mis on graphs $G$ of maximum degree $\Delta$
        within a factor of $\frac{\Delta}{\polylog(\Delta)}$.
        We shall apply \cref{obs:intro:hard} with $k = \Theta_\Delta(n)$;
        however, this makes the degree of the bipartite graph $K$
        (constructed in the above reduction) too large.
        To avoid this issue,
        we replace $K$ by a $\Delta$-regular \emph{bipartite Ramanujan graph} $X$.
        Informally speaking, bipartite Ramanujan graphs behave like random bipartite graphs
        (see \cref{sec:hardness:degree} for the definition).
        By the bipartite expander mixing lemma \cite{haemers1979eigenvalue,haemers1995interlacing},
        we show that $X$ does not contain independent sets of size $\Omega\left(\frac{k}{\sqrt{\Delta}}\right)$,
        which allows us to create a $\sqrt{\Delta}$-factor gap.

    \item[(Bandwidth $n^{\frac{1}{2}+\delta}$)]
        By \cite{MR1687331}, we first show that
        for an $n$-vertex graph $G$ of bandwidth $n^{\frac{1}{2}+\delta}$,
        it is $\NP$-hard to approximate $\alpha(G)$ within a factor of $\approx n^{4\delta^2}$.
        By applying \cref{obs:intro:hard}
        with $k \approx n^{\frac{1}{2}+\delta}$,
        we obtain the $\NP$-hardness of approximating \mmisr on graphs of bandwidth $n^{\frac{1}{2}+\delta}$
        within the same factor.

    \item[(Bipartite graphs)]
        Instead of using the above reduction, we show the ``equivalence'' (up to a constant factor) between
        \textsc{Maximum Balanced Biclique} and
        \mmisr on bipartite graphs.
        See \cref{sec:hardness:bipartite} for the formal definition of \textsc{Maximum Balanced Biclique}.
        Since this problem cannot be approximated within a factor of $n^{1-\epsilon}$
        under SSEH and $\NP \nsubseteq \BPP$ \cite{Manurangsi18biclique},
        where $n$ is the number of vertices in an input graph,
        \mmisr cannot be approximated within the same factor.
\end{description}

\section{Approximation algorithms for general graphs}
\label{sec:algo:general}
In this section, we present a polynomial-time $\frac{n}{\log n}$-factor approximation algorithm for \mmisr.

\begin{theorem}
    \label{thm:general_aprox}
    There exists a polynomial-time algorithm that, given an $n$-vertex graph $G$
    and two independent sets $\iini$ and $\itar$ of $G$,
    outputs an $(\iini,\itar)$-reconfiguration sequence $\mathscr{I}$ satisfying
    \begin{align*}
        \val_G(\mathscr{I}) \geq \frac{\log n}{n} \cdot \opt_G(\iini \reco \itar).
    \end{align*}
\end{theorem}

We begin with two simple but useful observations.

\begin{observation}
    \label{lem:pigeonhole}
    Let $V$ be a finite set partitioned into $V_1, V_2, \ldots, V_{\ell}$.
    For any subset $S \subseteq V$, there exists an index $i \in [\ell]$ such that $|S \cap V_i| \geq \frac{|S|}{\ell}$.
\end{observation}

\begin{proof}
    Suppose, for a contradiction, that $|S \cap V_i| < \frac{|S|}{\ell}$ holds for all $i \in [\ell]$.
    Then, $|S| = \sum_{i = 1}^{\ell} |S \cap V_i| < \ell \cdot \frac{|S|}{ \ell} = |S|$,
    which is a contradiction.
\end{proof}

\begin{observation}
    \label{lem:subset_incl}
    Let $I_1$ and $I_2$ be finite sets such that either $I_1 \subseteq I_2$ or $I_2 \subseteq I_1$.
    For any subsets $J_1 \subseteq I_1$ and $J_2 \subseteq I_2$,
    either $J_1 \subseteq I_2$ or $J_2 \subseteq I_1$ holds.
\end{observation}

\begin{proof}
    If $I_1 \subseteq I_2$,
    then $J_1 \subseteq I_1 \subseteq I_2$.
    Otherwise,
    $J_2 \subseteq I_2 \subseteq I_1$.
\end{proof}

Let $I = (G, \iini, \itar)$ be an instance of \mmisr.
For a nonnegative integer $\gamma$, a sequence $(J^{(1)}, J^{(2)}, \ldots, J^{(t)})$ of independent sets of $G$ is called a \defi{$\gamma$-sequence} with respect to $\iini$ and $\itar$ if it satisfies the following conditions:
\begin{description}
    \item[S1] $J^{(1)} \subseteq \iini$ and $J^{(t)} \subseteq \itar$,
    \item[S2] $|J^{(i)}| \geq \gamma$ for all $i \in [t]$, and
    \item[S3] $J^{(i)} \cup J^{(i+1)}$ is an independent set of $G$
    for all $i \in [t-1]$.
\end{description}

We show that the existence of a $\gamma$-sequence guarantees the existence of an $(\iini,\itar)$-reconfiguration sequence whose value is at least~$\gamma$.

\begin{lemma}
    \label{lem:ganma_set_to_reconf_seq}
    Let $I = (G, \iini, \itar)$ be an instance of \mmisr,
    and let $\gamma$ be a nonnegative integer.
    Given a $\gamma$-sequence $\mathscr{J} = (J^{(1)}, J^{(2)}, \ldots, J^{(t)})$ of length $t$,
    one can compute an $(\iini,\itar)$-reconfiguration sequence $\mathscr{S}$ such that $\val_G(\mathscr{S}) \geq \gamma$ in time polynomial in $|V(G)|$ and~$t$.
\end{lemma}

\begin{proof}
    From $\mathscr{J}$, we construct a reconfiguration sequence by the following procedure:
    \begin{description}
        \item[Step 1.]
        Remove all vertices in $\iini \setminus J^{(1)}$
        to obtain $J^{(1)}$.
        \item[Step 2.]
        For each $i = 1, 2, \ldots, t-1$,
        add all vertices in $J^{(i+1)} \setminus J^{(i)}$ to $J^{(i)}$
        to obtain $J^{(i)} \cup J^{(i+1)}$,
        and then remove all vertices in $J^{(i)} \setminus J^{(i+1)}$ from $J^{(i)} \cup J^{(i+1)}$ 
        to obtain $J^{(i+1)}$.
        \item[Step 3.] Add all vertices in $\itar \setminus J^{(t)}$ to $J^{(t)}$ 
        to obtain $\itar$.
    \end{description}
    Note that each operation adds or removes exactly one vertex. 
    Each step can be done in time polynomial in $|V(G)|$,
    and hence the total running time is polynomial in $|V(G)|$ and $t$.

    Since $\iini$ and $\itar$ are independent sets of $G$,
    all intermediate sets appearing in Steps~1 and~3 are independent sets.
    Moreover, since $J^{(i)} \cup J^{(i+1)}$ is an independent set of $G$ for every $i \in [t-1]$,
    all sets appearing in Step~2 are also independent sets.
    Therefore, the obtained sequence $\mathscr{S}$ is an $(\iini,\itar)$-reconfiguration sequence.
    
    Since $|J^{(i)}| \geq \gamma$ for every $i \in [t]$, every intermediate set in $\mathscr{S}$ has size at least $\gamma$.
    Thus, $\val_G(\mathscr{S}) \geq \gamma$, as claimed.
\end{proof}

We now show how to efficiently construct a $\gamma$-sequence with sufficiently large $\gamma$.

\begin{lemma}
    \label{lem:computing_ganma_seq}
    Let $I = (G, \iini, \itar)$ be an instance of \mmisr,
    where $G$ has $n$ vertices.
    There exists a polynomial-time algorithm that finds a $\gamma$-sequence of length polynomial in $n$ satisfying
    \begin{align*}
        \gamma \geq \frac{\log n}{n} \cdot \opt_G(\iini \reco \itar).
    \end{align*}
\end{lemma}

\begin{proof}
    Without loss of generality, assume that $n \geq 2$.
    We first prove the existence of such a $\gamma$-sequence.

    Let $\ell = \floor{n / \log n}$.
    We fix an arbitrary partition of the vertex set $V(G)$ into $V_1 \cup V_2 \cup \cdots \cup V_{\ell}$
    such that $\floor{\log n} \leq |V_i| \leq 2\ceil{\log n}$ for all $i \in [\ell]$.
    Such a partition exists since $\floor{\log n} \cdot \ell \leq n \leq 2\ceil{\log n} \cdot \ell$ for $n \geq 2$.

    Let $\mathscr{S} = (I^{(0)}, I^{(1)}, \ldots, I^{(t)})$ be an $(\iini,\itar)$-reconfiguration sequence attaining the optimum, that is, $\val_G(\mathscr{S}) = \opt_G(\iini \reco \itar)$.
    Define $\gamma \coloneq \frac{\val_G(\mathscr{S})}{\ell}\geq \frac{\log n}{n} \cdot \val_G(\mathscr{S})$.

    By \Cref{lem:pigeonhole},
    for each $j \in [t]$, there exists an index $\pi(j) \in [\ell]$ such that
    \begin{align}
        |I^{(j)} \cap V_{\pi(j)}| \geq\gamma .
        \label{ineq:set_size_constraint}
    \end{align}
    Define $J^{(j)} \coloneq I^{(j)} \cap V_{\pi(j)}$ for each $j \in [t]$.
    Clearly, each $J^{(j)}$ is an independent set of $G$.

    We claim that $\mathscr{J} = (J^{(1)}, J^{(2)}, \ldots, J^{(t)})$ is a $\gamma$-sequence.
    Condition~(S1) holds since $J^{(1)} \subseteq \iini$ and $J^{(t)} \subseteq \itar$.
    Condition~(S2) follows immediately from \Cref{ineq:set_size_constraint}.
    We verify Condition~(S3).
    Since $I^{(j)}$ and $I^{(j+1)}$ are two consecutive elements of a reconfiguration sequence, either $I^{(j)} \subseteq I^{(j+1)}$ or $I^{(j)} \supseteq I^{(j+1)}$.
    As $J^{(j)} \subseteq I^{(j)}$ and $J^{(j+1)} \subseteq I^{(j+1)}$,
    \Cref{lem:subset_incl} implies that either $J^{(j+1)} \subseteq I^{(j)}$ or $J^{(j)} \subseteq I^{(j+1)}$.
    Hence, $J^{(j)} \cup J^{(j+1)}$ is contained in either $I^{(j)}$ or $I^{(j+1)}$,
    and hence forms an independent set of $G$.
    Therefore, $\mathscr{J}$ is a $\gamma$-sequence.

    We next bound the length of a $\gamma$-sequence by modifying $\mathscr{J}$.
    For a sequence $\mathscr{X}' = (X^{(1)},X^{(2)}, \ldots, X^{(t)})$ of independent sets of $G$,
    a pair $(p,q)$ with $1 \leq p < q \leq t$ is called \defi{redundant} if $X^{(p)} = X^{(q)}$.
    For each redundant pair $(p,q)$ for $\mathscr{X}$,
    we replace the substring $(X^{(p)}, X^{(p+1)}, \ldots, X^{(q)})$ by the single set $X^{(p)}$.
    This operation preserves the property of being a $\gamma$-sequence.

    Since $|V_i| = O(\log n)$ for each $i \in [\ell]$,
    the number of subsets of $V_i$ is $O(n)$.
    Thus, the total number of distinct independent sets that may appear in $\mathscr{J}$ is at most $\sum_{i=1}^{\ell} 2^{|V_i|} = O(\ell \cdot n) = O(n^2 / \log n)$.
    After exhaustively processing redundant pairs, we obtain a $\gamma$-sequence where every element is unique. 
    The length of such a sequence is bounded by $O(n^2 / \log n)$.
    This establishes the existence of a $\gamma$-sequence with the desired properties.

    We finally describe a polynomial-time algorithm for finding such a sequence.
    Fix the above partition of $V(G)$.
    For each integer $\gamma \in \{0,1,\ldots,\lceil \log n \rceil\}$,
    we construct an auxiliary graph $H_\gamma$.
    The vertex set of $H_\gamma$ consists of all independent sets $I$ of $G$ such that $|I| \geq\gamma$ and $I \subseteq V_i$ for some $i \in [\ell]$.
    Two vertices $I$ and $J$ of $H_\gamma$ are adjacent if and only if $I \cup J$ is an independent set of $G$.
    Since $|V_i| = O(\log n)$, we have $|V(H_\gamma)| = O(n^2 / \log n)$.

    For a given $\gamma$, we test whether there exist independent sets $J_{\mathrm{ini}} \subseteq \iini$ and $J_{\mathrm{tar}} \subseteq \itar$ of size at least $\gamma$ such that there is a path between $J_{\mathrm{ini}}$ and $J_{\mathrm{tar}}$ in $H_\gamma$.
    Note that connectivity in $H_\gamma$ can be tested in polynomial time.
    If such a path exists, the sequence of independent sets along the path forms a $\gamma$-sequence.
    Conversely, the $\gamma$-sequence obtained in the above construction corresponds to a path in $H_\gamma$.
    Thus, for a fixed $\gamma$, the existence of a $\gamma$-sequence can be decided in polynomial time.

    Therefore, we can compute the maximum $\gamma \in \{0,1,\ldots,\lceil \log n \rceil\}$ for which a $\gamma$-sequence exists.
    Since the algorithm returns a $\gamma$-sequence with $\gamma \geq\frac{\log n}{n} \cdot \val_G(\mathscr{S})$, the approximation factor is at least $\frac{\log n}{n}$.
\end{proof}

By combining \Cref{lem:ganma_set_to_reconf_seq,lem:computing_ganma_seq}, we obtain \Cref{thm:general_aprox}.

\section{Approximation algorithms for special classes of graphs}
\label{sec:algo}

In this section, we present approximation algorithms for graphs under different restrictions, specifically, degenerate graphs, graphs with bounded treewidth, and $H$-minor-free graphs.

\subsection{Degenerate graphs}
In this subsection, we show the following.

\begin{theorem}
\label{thm:algo:degeneracy}
    There exists a polynomial-time algorithm that, given a $d$-degenerate graph $G$ and two independent sets $\iini$ and $\itar$ of $G$, outputs an $(\iini, \itar)$-reconfiguration sequence $\mathscr{J}$ such that 
    \[
        \val_G(\mathscr{J}) \geq \frac{1}{d} \cdot \opt_G(\iini \reco \itar)-1.
    \]
\end{theorem}
\begin{proof}
For two vertex subsets $A, B$ of $V(G)$, let $G[A,B]$ denote the bipartite subgraph of $G$ induced by the edge set between $A$ and $B$.
If $A$ and $B$ are independent sets, then $G[A,B]$ coincides with the induced subgraph $G[A \cup B]$.

\subparagraph{Algorithm.}
In our algorithm, we maintain three variables $X$, $Y$, and $Z$, which store independent sets of $G$.
Initially, we set $Z \coloneq \iini \cap \itar$, $X \coloneq \iini \setminus Z$, and $Y \coloneq \itar \setminus Z$.
Then we repeat the following procedure until $X = Y = \emptyset$.
Let $v$ be a vertex with the smallest degree in the bipartite graph $G[X,Y]$.
Let $T \in \{X, Y\}$ be the set containing $v$, and let $T'$ denote the other set.
Moreover, let $N_{X,Y}(v)$ be the set of neighbours of $v$ in $G[X,Y]$.
We then remove $v$ from $T$, remove all vertices in $N_{X,Y}(v)$ from $T'$, and add $v$ to $Z$.

Note that, throughout this procedure, the sets $X$ and $Y$ are maintained to be independent sets of $G$.
Consequently, the induced subgraph $G[X \cup Y]$ remains bipartite.

Suppose that the above procedure is executed for $t-1$ iterations.
Let $(X^{(1)}, X^{(2)}, \dots, X^{(t)} = \emptyset)$, $(Y^{(1)}, Y^{(2)}, \dots, \allowbreak Y^{(t)} = \emptyset)$, and $(Z^{(1)}, Z^{(2)}, \dots, Z^{(t)})$ denote the sequences of sets assigned to $X$, $Y$, and $Z$ during the procedure, respectively.
For $i \in [t-1]$, we define 
\begin{align*}
\overrightarrow{J}^{(2i-1)} \coloneq X^{(i)} \cup Z^{(i)}, \quad &\overrightarrow{J}^{(2i)} \coloneq X^{(i+1)} \cup Z^{(i)},\\ 
\overleftarrow{J}^{(2i-1)} \coloneq Y^{(i)} \cup Z^{(i)}, \quad&\overleftarrow{J}^{(2i)} \coloneq Y^{(i+1)} \cup Z^{(i)}.    
\end{align*}
Finally, we set $\overrightarrow{J}^{(2t-1)} \coloneq X^{(t)} \cup Z^{(t)} = Y^{(t)} \cup Z^{(t)} \eqqcolon \overleftarrow{J}^{(2t-1)}$.
The algorithm outputs the sequence $\mathscr{J} \coloneq \left(\overrightarrow{J}^{(1)}, \overrightarrow{J}^{(2)}, \dots, \overrightarrow{J}^{(2t-1)} = \overleftarrow{J}^{(2t-1)}, \dots,\overleftarrow{J}^{(2)}, \overleftarrow{J}^{(1)}\right)$.

\subparagraph{Correctness.}
We first prove that the first $2t-1$ elements of $\mathscr{J}$ form a reconfiguration sequence; that is, each element is an independent set of $G$, and for every pair of consecutive elements $A$ and $B$, either $A \subseteq B$ or $A \supseteq B$ holds.
The proof proceeds by induction along the sequence.

Observe that $\overrightarrow{J}^{(1)} = \iini$, which is an independent set of $G$.
Suppose that for some $i \in [t-1]$, the set $\overrightarrow{J}^{(2i-1)}$ is an independent set of $G$.
Consider three consecutive elements $\overrightarrow{J}^{(2i-1)} = X^{(i)} \cup Z^{(i)}$, $\overrightarrow{J}^{(2i)} = X^{(i+1)} \cup Z^{(i)}$, and $\overrightarrow{J}^{(2i+1)} = X^{(i+1)} \cup Z^{(i+1)}$.
By construction, there is a vertex $v$ such that $Z^{(i+1)} = Z^{(i)} \cup \{v\}$ and either 
\begin{enumerate}
    \item[(i)] $X^{(i+1)} = X^{(i)} \setminus \{v\}$ with $v \in X^{(i)}$, or
    \item[(ii)] $X^{(i+1)} = X^{(i)} \setminus N_{X^{(i)},Y^{(i)}}(v)$ with
    $v \in Y^{(i)}$.
\end{enumerate}
In both cases, we have $\overrightarrow{J}^{(2i-1)} \supseteq \overrightarrow{J}^{(2i)}$ and $\overrightarrow{J}^{(2i)} \subseteq \overrightarrow{J}^{(2i+1)}$.
Since $\overrightarrow{J}^{(2i-1)}$ is an independent set, it follows that $\overrightarrow{J}^{(2i)}$ is also an independent set.

It therefore remains to show that $\overrightarrow{J}^{(2i+1)}$ is an independent set.
In case~(i), we have $\overrightarrow{J}^{(2i+1)} = \overrightarrow{J}^{(2i-1)}$, and the claim follows immediately.
In case~(ii), by construction, $N_{X^{(i)},Y^{(i)}}(v)$ contains all neighbours of $v$ in $X^{(i)}$, and hence $v$ has no neighbours in $X^{(i+1)}$.
Moreover, when a vertex $w$ is added to $Z$, it has no neighbours in the current sets $X$ and $Y$.
Therefore, no vertex in $Z^{(i)}$ is adjacent to $v$, and since $\overrightarrow{J}^{(2i)}$ is an independent set, we conclude that $\overrightarrow{J}^{(2i+1)}$ is an independent set as well.
This completes the inductive argument.

By a symmetric argument, the last $2t-1$ elements of $\mathscr{J}$ also form a reconfiguration sequence.
Together with the facts that $\overrightarrow{J}^{(1)} = \iini$ and $\overleftarrow{J}^{(1)} = \itar$, we conclude that $\mathscr{J}$ is an $(\iini, \itar)$-reconfiguration sequence as required.

It remains to analyse the approximation factor.
Observe that since $G$ is $d$-degenerate, the chosen vertex $v$ has degree at most $d$ in $G[X, Y]$.
In each iteration $i = 1,2,\dots,t-1$, the algorithm removes at most $d$ vertices from $X^{(i)}$ and adds exactly one vertex to $Z^{(i)}$.
Hence, the size of $X^{(i)} \cup Z^{(i)}$ decreases by at most $d-1$ per iteration.
On the other hand, after $i$ iterations, exactly $i$ vertices have been added to $Z^{(i+1)}$, and therefore $|X^{(i+1)} \cup Z^{(i+1)}| \ge i$.
Combining these observations, after $i$ iterations, we obtain 
\[
|X^{(i+1)} \cup Z^{(i+1)}| \geq \max\{i,\ |\iini| - (d-1) \cdot i \} \ge \frac{|\iini|}{d}.
\]
In particular, for $i \in [t]$, it holds that $|\overrightarrow{J}^{(2i-1)}| \geq \frac{|\iini|}{d}$.

Recall that, for each $i\in [t-1]$, we have $|\overrightarrow{J}^{(2i)}| = |\overrightarrow{J}^{(2i+1)}|-1$. 
Consequently, every independent set $\overrightarrow{J}^{(1)}, \overrightarrow{J}^{(2)}, \dots, \overrightarrow{J}^{(2t-1)}$ has size at least $\frac{|\iini|}{d} - 1$.

By symmetry, the same argument implies that every independent set $\overleftarrow{J}^{(1)}, \overleftarrow{J}^{(2)}, \dots, \overleftarrow{J}^{(2t-1)}$ has size at least $\frac{|\itar|}{d} - 1$.

Finally, since $\opt_G(\iini \reco \itar) \le \min\{|\iini|,|\itar|\}$, we conclude that
\[
    \val_G(\mathscr{J}) \ge \min\left\{\frac{|\iini|}{d}-1,\ \frac{|\itar|}{d}-1\right\} \ge \frac{1}{d}\,\opt_G(\iini \reco \itar) - 1,
\]
which completes the correctness proof of \Cref{thm:algo:degeneracy}.

\subparagraph{Run time.} 
The number of iterations in the procedure above is at most $n$.
In each iteration, a vertex of minimum degree in $G[X,Y]$ can be found in $\bigO(nd)$ time, and the sets $X$, $Y$, and $Z$ can be updated in linear time.
Therefore, the overall run time is polynomial.
\end{proof}

\subsection{Bounded treewidth}
In this subsection, we discuss several algorithms for graphs with bounded treewidth.
Recall that a graph $G$ with treewidth $\tw(G)$ is $\tw(G)$-degenerate.
Therefore, the algorithm presented in the previous subsection applies directly to this setting.
However, by utilising the existence of a small balanced separation, we can design an alternative algorithm that achieves a better approximation factor in certain cases.
Moreover, we can combine this algorithm with an \fpt algorithm for \mmisr to devise an \fpt-approximation scheme for the problem.

Throughout this subsection, let $(G, \iini, \itar)$ be an instance of \mmisr.
For brevity, we define $\imin \coloneq \min\{|\iini|, |\itar|\}$.

\subsubsection{A polynomial-time approximation algorithm}
The main result of this section is the following approximation algorithm.

\begin{theorem}
\label{thm:algo:treewidth}
    Let $t$ be a positive integer.
    There exists a polynomial-time algorithm that, given a graph $G$ together with a tree decomposition of width $t-1$ and two independent sets $\iini$ and $\itar$ of $G$, outputs an $(\iini, \itar)$-reconfiguration sequence $\mathscr{J}$ such that 
    \[
        \val_G(\mathscr{J}) \geq \imin - t \left(\log_{3/2} \frac{\imin}{t} +1 \right).
    \]    
\end{theorem}

The algorithm runs in polynomial time if the tree decomposition is given.
Otherwise, we additionally compute a tree decomposition with width $\tw(G)$ in \fpt time~\cite{Bodlaender96}.

Note that the lower bound in \cref{thm:algo:treewidth} is positive when $\imin$ is at least roughly $5t$.
(For example, it is clearly positive when $\imin \geq (3/2)^{4} t$.)
However, when $t$ is large relative to $\imin$, we can instead use the algorithm in \cref{thm:algo:degeneracy}.
As discussed before, a graph $G$ with treewidth $t-1$ is $(t-1)$-degenerate.
Hence, simply by combining \cref{thm:algo:degeneracy} and \cref{thm:algo:treewidth}, we can obtain a better algorithm, as follows.

\begin{corollary}
\label{cor:treewidth}
    Let $t$ be a positive integer.
    There exists a polynomial-time algorithm that, given a graph $G$ together with a tree decomposition of width $t-1$ and two independent sets $\iini$ and $\itar$ of $G$, outputs an $(\iini, \itar)$-reconfiguration sequence $\mathscr{J}$ such that 
    \[
        \val_G(\mathscr{J}) \geq \max\left\{\frac{\imin - t \left(\log_{3/2} (\imin/t) +1 \right)}{\imin} \cdot\opt_G(\iini \reco \itar), 
        \frac{1}{t-1} \cdot \opt_G(\iini \reco \itar) - 1
        \right\}.
    \]
\end{corollary}

The proof of \cref{thm:algo:treewidth} relies on the existence of a so-called balanced $\iini$-separator.
A pair of vertex subsets $(A, B)$ is a \defi{separation} of $G$ if $A \cup B = V(G)$ and there is no edge between any vertex in $A \setminus B$ and any vertex in $B \setminus A$.
For a nonnegative weight function $w\colon V(G) \to \mathbb{R}_{\ge 0}$ and a set $V'\subseteq V(G)$, we define $w(V') = \sum_{v \in V'} w(v)$.
For a constant $\alpha \in (0, 1)$, we say that a separation $(A, B)$ of $G$ is an \defi{$\alpha$-balanced separation} if $w(A\setminus B) \le \alpha \cdot w(V(G))$ and $w(B\setminus A) \le \alpha \cdot w(V(G))$.

\begin{lemma}[\cite{CyganFKLMPPS15book}] \label{lem:tw_balance}
    Let $G$ be a graph with a nonnegative weight function $w\colon V(G) \to \mathbb{R}_{\ge 0}$.
    Then there exists a $\frac{2}{3}$-balanced separation $(A, B)$ of $G$ such that $|A\cap B| \le \tw(G)+1$.
    Moreover, given a tree decomposition of $G$ of width $\tw(G)$, such a separation can be found in polynomial time.
\end{lemma}

An \defi{$\alpha$-balanced $\iini$-separation} of $G$ is a separation $(A,B)$ such that  $|\iini \cap (A\setminus B)| \le \alpha \cdot |\iini|$, and $|\iini \cap (B\setminus A)| \le \alpha \cdot |\iini|$.
The intersection $A \cap B$ is usually called an \defi{$\alpha$-balanced $\iini$-separator}.
The following corollary is obtained by applying \cref{lem:tw_balance} with a weight function $w$ such that $w(v) = 1$ if $v \in \iini$ and $w(v) = 0$ otherwise.

\begin{corollary} \label{cor:tw_inibalance}
    There exists a $\frac{2}{3}$-balanced $\iini$-separation $(A, B)$ in $G$ such that $|A\cap B| \le \tw(G)+1$.
    Moreover, given a tree decomposition of $G$ of width $\tw(G)$, such a separation can be found in polynomial time.
\end{corollary}

We are now ready to prove \cref{thm:algo:treewidth}.

\begin{proof}[Proof of \cref{thm:algo:treewidth}]
    In the following, we describe a recursive algorithm, denoted by Algorithm~$T$, that takes as input a graph $G$, a tree decomposition of $G$ with width $t-1$, and two independent sets $\iini$ and $\itar$ of $G$. 
    The algorithm constructs an $(\iini, \itar)$-reconfiguration sequence by maintaining a current independent set $I$, initialised as $\iini$ and specifying the modifications on $I$ (i.e., adding or removing vertices).

    Algorithm $T$ is as follows:
    \begin{description}
        \item[Step 1.]
        If $\iini = \emptyset$,
        we simply add all $\itar$ to $I$ and stop processing further.
        \item[Step 2.] 
        Without loss of generality, we assume $|\iini| \leq |\itar|$; otherwise, in this and subsequent steps, we swap the roles of $\iini$ and $\itar$.
        Compute a $\frac{2}{3}$-balanced $\iini$-separation $(A, B)$ of $G$. Define $S \coloneq A \cap B$, $X \coloneq A \setminus S$, $Y \coloneq B \setminus S$, $\iini^X \coloneq \iini \cap X$, $\iini^Y \coloneq \iini \cap Y$, $\itar^X \coloneq \itar \cap X$, and $\itar^Y \coloneq \itar \cap Y$.
        \item[Step 3.] Remove $S \cap \iini$ from $I$.
        \item[Step 4] If $|\itar^X| - |\iini^X| \geq |\itar^Y| - |\iini^Y|$, define $P \coloneq X$ and $Q \coloneq Y$.
            Otherwise, define $P \coloneq Y$ and $Q \coloneq X$.
        \item[Step 5.] Recurse on $P$. (That is, if $P \neq \emptyset$, we call Algorithm~$T$ on input $G[P]$, the tree decomposition restricted to $P$, $\iini^P$, and $\itar^P$. 
        Apply the sequence of modifications returned by this call to the current set $I$.
        If $P = \emptyset$, we do nothing.)
        \item[Step 6.] Recurse on $Q$.
        \item[Step 7.] Add $S \cap \itar$ to $I$.
    \end{description}

    \subparagraph{Correctness.}
    We first prove by induction that $I$ is always an independent set throughout the execution, and after Step 7, $I = \itar$.
    This is trivially true for the base case when $G$ has no vertices. 
    At Step 3, we only remove vertices, and hence, after the modification, $I$ is still an independent set.
    At Step 5, by recursion, throughout the modifications in $P$, $I \cap P$ is an independent set.
    Note that there is no change in $I \cap Q$ during this process.
    By the definition of a separation, there are no edges between $P$ and $Q$, and hence, $I$ remains an independent set throughout the recursion.
    A similar argument also applies to Step 6.
    After this step, we have $I \cap P = \itar^P$ and $I \cap Q = \itar^Q$.
    Hence, $I = \itar^P \cup \itar^Q = \itar \cap (X \cup Y)$.
    Therefore, after adding $S \cap \itar$ to $I$, we obtain $I = \itar$, which is an independent set.
    This completes the inductive proof.

    Next, we argue about the approximation guarantee.
    We denote by $\lambda(G, \iini, \itar)$ the value of the reconfiguration sequence obtained from Algorithm $T$ on the input $(G, \iini, \itar)$ and a tree decomposition of width $t-1$ of $G$.
    The approximation guarantee follows immediately from the following claim.
    
    \begin{claim}
        $\lambda(G, \iini, \itar) \geq \imin - t \left(\log_{3/2} \frac{\imin}{t} +1 \right).$
    \end{claim}
    \begin{claimproof}
    We prove by induction on the size of 
    $\iini$.
    Note that this is trivially true for the base case when 
    $|\iini| = 0$, as in this case, $\imin = 0$.
    For the inductive step, suppose that 
    $\iini$ is not empty.
    We show that during the execution of Algorithm~$T$,
    \begin{equation}
    \label{eq:tw_recursion_value}
        |I| \geq \imin - t \left(\log_{3/2} \frac{\imin}{t} +1 \right).
    \end{equation}

    Initially, $|I| = |\iini| \ge \imin$, satisfying the bound.
    In Step~3, we remove at most $|S| \leq t$ vertices.
    Thus, $|I| =  |\iini \setminus S| \geq \imin - t$, which satisfies \Cref{eq:tw_recursion_value}.

    Next, we consider the recursive call on $P$ (Step~5).
    Define $\imin^P \coloneq \min\{|\iini^P|, |\itar^P|\}$.
    Firstly, as argued above, at the beginning of the step, we have $I = \iini^P \cup \iini^Q$.    
    Secondly, as per the algorithm's description, $I \cap Q$ does not change during this step.
    Thirdly, by the definition of a $\frac{2}{3}$-balanced $\iini$-separation, we have $\frac{2}{3}\imin = \frac{2}{3}|\iini| \geq |\iini^P| \geq \imin^P$.
    Fourthly, by the inductive hypothesis, we assume that 
    \begin{equation*}
        \lambda(P, \iini^P, \itar^P) \geq \imin^P -  t \left(\log_{3/2} \frac{\imin^P}{t} +1 \right) = \imin^P -  t \log_{3/2} \frac{3\imin^P}{2t}.
    \end{equation*}
    Lastly, by the definitions of $P$ and $Q$, we have $|\itar^P| - |\iini^P| \geq |\itar^Q| - |\iini^Q|$, or equivalently, $|\itar^P| + |\iini^Q| \geq |\itar^Q| + |\iini^P|$.
    Combining this with the facts that $|\iini| = |\iini^P| + |\iini^Q| +|\iini \cap S|$ and $|\itar| = |\itar^P| + |\itar^Q| +|\itar \cap S|$, we have
    \begin{equation}
    \label{eq:PQ}
        |\itar^P| + |\iini^Q| \geq \frac{1}{2}\left( |\iini^P| + |\iini^Q| + |\itar^P| + |\itar^Q| \right)
        = \frac{1}{2}\left( |\iini| + |\itar| - 2|S| \right) \geq \imin - t.
    \end{equation}
    This implies that if $\imin^P = |\itar^P|$, then we have $\imin^P + |\iini^Q| \geq \imin - t$.
    Otherwise, $\imin^P = |\iini^P|$, and hence, $\imin^P + |\iini^Q| = |\iini^P| + |\iini^Q| = |\iini \setminus S| \geq \imin - t$.
    Combining the preceding two sentences, we have $\imin^P + |\iini^Q| \geq \imin - t$ in all cases.

    Putting all of the above together, during Step 5, we always have
    \begin{align*}
        |I| &\geq |\iini^Q| + \lambda(P, \iini^P, \itar^P) \geq |\iini^Q| + \imin^P -  t \log_{3/2} \frac{3\imin^P}{2t} \\
        &\geq \imin - t - t \log_{3/2} \frac{3\imin^P}{2t} \geq \imin - t - t \log_{3/2} \frac{\imin}{t}.
    \end{align*}
    Hence, \cref{eq:tw_recursion_value} holds during Step 5.

    Now, at the beginning of Step 6, we have $I = \itar^P \cup \iini^Q$, which is at least $\imin - t$, by \cref{eq:PQ}.
    Hence, we can follow a similar analysis as in Step 5 and obtain \cref{eq:tw_recursion_value} throughout Step 6.

    Lastly, for Step 7, since we add vertices to $I$, its size can only increase, and hence \cref{eq:tw_recursion_value} continues to hold.    
    \end{claimproof}
    
    \subparagraph{Run time.}
    It is easy to see that Steps 1--4 and 7 take polynomial time.
    Since $|\iini \cap P| \leq \frac{2}{3} |\iini|$ and $|\iini \cap Q| \leq \frac{2}{3} |\iini|$, the recursion depth is $\bigO(\log |\iini|)$.
    Hence, overall, Algorithm $T$ runs in polynomial time.
\end{proof}

\subsubsection{\texorpdfstring{\fpt}-approximation scheme for bounded-treewidth graphs}
\label{sec:algo:tw:fpt}
This subsection is dedicated to proving the following.

\begin{theorem}
    \label{thm:treewidth_PTAS}
    Let $\varepsilon$ be a positive real number and $t$ a positive integer.
    Let $(G, \iini, \itar)$ be an instance of \mmisr with $|V(G)| = n$.
    Then there exists an algorithm that 
    outputs an $(\iini,\itar)$-reconfiguration sequence $\mathscr{J}$ in time
    $f(\tw(G), \varepsilon)\cdot n^{\bigO(1)}$ for some computable function $f$ such that
    \[
        \val_G(\mathscr{J}) \geq \frac{1}{1+\varepsilon}\,\opt_G(\iini \reco \itar).
    \]
\end{theorem}

The main idea is to combine \cref{thm:algo:treewidth} and the following theorem.

\begin{theorem}
\label{thm:FPT_degeneracy}
    Let $G$ be a $d$-degenerate $n$-vertex graph, and $\iini, \itar$ two independent sets of $G$ of size at least $\imin$.
    Then there exists an algorithm that outputs an $(\iini,\itar)$-reconfiguration sequence with value $\opt_G(\iini \reco \itar)$ that runs in time $f(d, \imin) \cdot n^{\bigO(1)}$ for some computable function $f$.
\end{theorem}

\cref{thm:FPT_degeneracy} above in turn relies on the equivalence between the Token Jumping and Token Addition/Removal rule, which we explain below.

\subparagraph{Equivalence between \mmisr and ISR-TJ.}
Kami\'{n}ski, Medvedev, and Milani\v{c}~\cite{kaminski2012complexity} showed the following equivalence between the Token Jumping rule and the Token Addition/Removal rule of ISR, which we phrase as follows.

\begin{theorem}[\cite{kaminski2012complexity}]
\label{thm:isr_tj_tar_equiv}
    Let $G$ be a graph, and let $\iini, \itar$ be two independent sets of $G$ with $|\iini| = |\itar| = \imin$.
    Then there exists a TJ-reconfiguration sequence from $\iini$ to $\itar$, if and only if $\opt_G(\iini \reco \itar) \geq \imin -1$.
\end{theorem}

We can extend the result above to a reduction between ISR-TJ and \mmisr in the sense that if we can solve one problem, then we can solve the other with a similar run time.
\cref{thm:isr_tj_tar_equiv} directly implies that an algorithm to solve \mmisr can be used to decide ISR-TJ.
For the other direction, we use the following lemma.

\begin{lemma}
\label{claim:subsets_opt}
    Let $G$ be a graph, and let $I, J, I', J'$ be its independent sets such that $I \subseteq I'$ and $J \subseteq J'$.
    Then $\opt_G(I \reco J) = \min\{|I|, |J|, \opt_G(I' \reco J')\}.$
\end{lemma}
\begin{proof}
    On the one hand, any $(I',J')$-reconfiguration sequence can be turned into a $(I, J)$-reconfiguration sequence by appending $I$ at the beginning and $J$ at the end.
    Therefore, 
    \begin{equation*}
        \opt_G(I \reco J) \geq \min\{|I|, |J|, \opt_G(I' \reco J')\}.
    \end{equation*}
    On the other hand, any $(I, J)$-reconfiguration sequence can be turned into a $(I',J')$-reconfiguration sequence by appending $I'$ at the beginning and $J'$ at the end.
    Hence, 
    \begin{equation*}
        \opt_G(I' \reco J') \geq \opt_G(I \reco J).
    \end{equation*}
    The lemma then follows from the two expressions above and the fact that $\min\{|I|,|J|\} \geq \opt_G(I \reco J)$.
\end{proof}

\cref{thm:isr_tj_tar_equiv} and \cref{claim:subsets_opt} imply that using an algorithm for ISR-TJ as a subroutine, we can solve \mmisr in a binary search fashion, as follows.

\begin{corollary}
\label{cor:equiv}
    Let $G$ be a graph, and let $\iini$ and $\itar$ be two independent sets of $G$.
    Define $\imin \coloneq \min \{ |\iini|, |\itar|\}$.
    Suppose there exists an algorithm~$A$ that solves ISR-TJ.
    Then there exists an algorithm~$B$ that solves the \mmisr instance $(G, \iini, \itar)$ by calling algorithm~$A$ $\bigO(\log \imin)$ times, with each call on an input of the form $(G, \iini',\itar')$ for some independent sets $\iini', \itar'$ of $G$ such that $|\iini'|, |\itar'| \leq \imin+ 1$.
    
    Moreover, if algorithm~$A$ additionally outputs a TJ-reconfiguration sequence in the case of a YES-instance, then algorithm~$B$ also outputs a $(\iini,\itar)$-reconfiguration sequence with value $\val_G(\iini \reco \itar)$.
\end{corollary}

\begin{proof}
    Without loss of generality, we assume that $|\iini| = |\itar| = \imin$.
    Indeed, if $|\iini| \neq |\itar|$, we can take subsets $\iini' \subseteq \iini$ and $\itar' \subseteq \itar$ such that $|\iini'| = |\itar'| = \min\{|\iini|, |\itar|\}$.
    It can be seen that $\opt_{G}(\iini \leftrightsquigarrow \itar) = \opt_{G}(\iini^\prime \leftrightsquigarrow \itar^\prime) $.
    
    Consider the algorithm $B$ described as follows.
    \begin{description}
        \item[Step 1.] Compute two arbitrary independent sets $I_1$ and $I_2$ such that $|I_1| = |I_2| = \imin+1$, $I_1 \supset \iini$, and $I_2 \supset \itar$.
        If no such sets exist, go to Step~3.
        \item[Step 2.] Call Algorithm~A on input $(G, I_1, I_2)$. If it outputs Yes, return $\imin$ and stop processing further.
        \item[Step 3.] Set $\ell \coloneq 1$, $r \coloneq \imin$, and $s \coloneq 0$. 
        \item[Step 4.] While $\ell < r$,
        \item[Step 4.1.] Let $p = \lceil \frac{\ell+r}{2} \rceil$.
        \item[Step 4.2] Compute two arbitrary subsets $I'_1$ and $I'_2$ of $\iini$ and $\itar$, respectively, such that $|I'_1| = |I'_2|= p$.
        \item[Step 4.3.] Call Algorithm~A on input $(G, I'_1, I'_2)$. If it outputs Yes, set $s \coloneq p-1$ and $\ell \coloneq p + 1$. Else, set $r \coloneq p - 1$.
        \item[Step 5] Return $s$.
    \end{description}

    It is easy to see that algorithm~$B$ above calls Algorithm~$A$ $\bigO(\log \imin)$ times.
    Hence, it remains to show its correctness.
    
    If algorithm~A outputs Yes in Step 2, then $\opt_G(I_1 \reco I_2) \geq \imin$.
    Combined with \cref{claim:subsets_opt}, this implies that $\opt_G(\iini \reco \itar) = \imin$, which is the value that the procedure returns.
    Otherwise, we then have $\opt_G(I_1 \reco I_2) < \imin$, and hence by \cref{claim:subsets_opt}, we obtain $\opt_G(\iini \reco \itar) < \imin$.
    
    Next, in Step 4.3, if algorithm~$A$ outputs No, then $\opt_G(I'_1 \reco I'_2) < p-1$.
    Combined with \cref{claim:subsets_opt}, we then have $\opt_G(\iini \reco \itar) < p-1$.
    Otherwise, if algorithm~$A$ outputs Yes at this step, then $\opt_G(I'_1 \reco I'_2) \geq p-1$, and again by \cref{claim:subsets_opt}, we conclude that $\opt_G(\iini \reco \itar) \geq p-1$.
    From these observations, combined with the fact that $0 \leq \opt_G(\iini \reco \itar) \leq \imin - 1$, we can easily conclude that the procedure correctly outputs $\opt_G(\iini \reco \itar)$ in Step 5.

    Finally, as observed in~\cite{kaminski2012complexity}, from a TJ-reconfiguration sequence $\mathscr{J}$ between two independent sets $I$ and $J$ of size $s$, we can construct an $(I,J)$-reconfiguration sequence of value $s-1$ with length $2|\mathscr{J}|-1$.
    We can do this by simply inserting $X \cap Y$ between every two consecutive elements $X$ and $Y$ of $\mathscr{J}$.
    Using this observation, it is easy to see that algorithm $B$ above can be modified to output an optimal reconfiguration sequence for the \mmisr problem when algorithm $A$ outputs a TJ-reconfiguration sequence in the case of a YES ISR-TJ instance.
\end{proof}

\subparagraph{\fpt algorithm for \mmisr by degeneracy.}
Besides the equivalence discussed above, another ingredient for the proof of \cref{thm:FPT_degeneracy} is the following recent result.

\begin{theorem}[\cite{AgrawalHM25}]\label{thm:ISR_FPT_degeneracy}
    Let $G$ be a $d$-degenerate $n$-vertex graph, and $\iini, \itar$ two independent sets of $G$ of size $\imin$.
    Then there exists an algorithm that decides the ISR-TJ instance $(G, \iini, \itar)$ in time $f(\imin, d) \cdot n^{\bigO(1)}$ for some computable function $f$.
    Moreover, if $(G, \iini, \itar)$ is a YES-instance, the algorithm also outputs a TJ-reconfiguration sequence from $\iini$ to $\itar$.
\end{theorem}

While \cref{thm:ISR_FPT_degeneracy} is not stated explicitly in \cite{AgrawalHM25}, it can easily be derived from Theorem~1.2 and the proof of Lemma~7.1 in the paper.

\begin{proof}[Proof of \cref{thm:FPT_degeneracy}]
    This theorem is a direct consequence of \cref{cor:equiv} and \cref{thm:ISR_FPT_degeneracy}.
\end{proof}

\subparagraph{\fpt-approximation scheme for \mmisr.}
We are now ready to prove \cref{thm:treewidth_PTAS}.

\begin{proof}[Proof of \cref{thm:treewidth_PTAS}]
    Define $t \coloneq \tw(G)+1$.
    If $0 < \frac{\imin}{\imin - t \left(\log_{3/2} (\imin/t) +1 \right)} \leq 1 + \varepsilon$, then we can use the algorithm in \cref{thm:algo:treewidth} to compute an $(\iini, \itar)$-reconfiguration sequence with value at least $\frac{1}{1+\varepsilon}\opt_G(\iini \reco \itar)$ in polynomial time, provided that the tree decomposition is given.
    If this tree decomposition is not given, we can precompute it in \fpt time~\cite{Bodlaender96}.
    
    Otherwise, we have $\frac{\imin}{\imin - t \left(\log_{3/2} (\imin/t) +1 \right)} > 1 + \varepsilon$ or $\imin - t \left(\log_{3/2} (\imin/t) +1 \right) \leq 0$. 
    In either case, $\imin$ is bounded from above by $h(t, \varepsilon)$ for some computable function $h$. 
    Since $\tw(G) = t-1$, $G$ is $(t-1)$-degenerate.
    \cref{thm:FPT_degeneracy} then implies
    that we can compute the optimal $(\iini,\itar)$-reconfiguration sequence in time $g(\varphi, t)\cdot n^{\bigO(1)}$ for some computable function $g$.
    Since $\varphi \leq h(t, \varepsilon)$, it means that this can be done in time $f'(\varphi, t) \cdot n^{\bigO(1)}$ for some computable function $f'$.
\end{proof}

\subsection{\texorpdfstring{$H$}-minor-free graphs}

In this subsection, we present an \fpt-approximation scheme for \mmisr on $H$-minor-free graphs.
Our main result is the following theorem.

\begin{theorem}
    \label{thm:minor_free}
    Let $\varepsilon > 0$ and let $H$ be a graph.
    There exists an algorithm that, given an $H$-minor-free graph $G$ on $n$ vertices and two independent sets $\iini$ and $\itar$ of $G$, outputs an $(\iini,\itar)$-reconfiguration sequence $\mathscr{J}$ in time
    $f(\varepsilon,|V(H)|)\cdot n^{O(1)}$ for some computable function $f$ such that
    \[
        \val_G(\mathscr{J}) \geq \frac{1}{1+\varepsilon}\,\opt_G(\iini \reco \itar).
    \]
\end{theorem}

The proof of \Cref{thm:minor_free} follows the framework underlying Baker's technique, which yields PTASs for a wide range of combinatorial optimisation problems on planar graphs~\cite{Baker94} and, more generally, on $H$-minor-free graphs~\cite{DemaineHK05bidimension}.
In particular, we rely on the following structural decomposition theorem.

\begin{theorem}[{\cite[Theorem~3.1]{DemaineHK05bidimension}}]
    \label{lem:partition_minor_free}
    For any graph $H$ on $h$ vertices, there exists a constant $c_h$, depending only on $h$, such that for every integer $k \geq 1$ and every $H$-minor-free graph $G$, the vertex set of $G$ can be partitioned into $k+1$ subsets $V_1, V_2, \ldots, V_{k+1}$ with the property that the subgraph induced by the union of any $k$ of these subsets has treewidth at most $c_h \cdot k$.
    Moreover, such a partition can be computed in polynomial time.
\end{theorem}

\begin{proof}[Proof of \cref{thm:minor_free}]
Let $k$ be a number that we will choose later.
Let $V_1,V_2,\ldots,V_{k+1}$ be a partition of $V(G)$ obtained by \Cref{lem:partition_minor_free}.
For each $j \in [k+1]$, let $G_j$ denote the subgraph of $G$ induced by $V(G)\setminus V_j$.

Without loss of generality, we assume that $|\iini| = |\itar|$. 
If this is not the case, let $\iini' \subseteq \iini$ and $\itar' \subseteq \itar$ be arbitrary subsets such that
$|\iini'| = |\itar'| = \min\{|\iini|, |\itar|\}$.
This reduction preserves the optimal value of the reconfiguration problem, that is, $\opt_{G}(\iini \leftrightsquigarrow \itar)
=
\opt_{G}(\iini' \leftrightsquigarrow \itar')$.

Let $\eta \coloneq |\iini| = |\itar|$.
Since $V_1, V_2, \ldots, V_{k+1}$ is a partition of $V(G)$, an argument analogous to the proof of \Cref{lem:pigeonhole} implies the existence of an index $j \in [k+1]$ such that 
\[
|(\iini \cup \itar) \cap V_j| \leq \frac{|\iini \cup \itar|}{k+1} \leq \frac{2\eta}{k+1}.
\]
Hence,
\[
\max\{|\iini \cap V_j|,|\itar \cap V_j| \}\leq |(\iini \cup \itar) \cap V_j| \leq \frac{2\eta}{k+1}.
\]
Therefore,
\[
|\iini \cap V(G_j)| \ge |\iini| - \frac{2\eta}{k+1} \ge \left(1 - \frac{2}{k+1}\right)\eta \quad \text{and} \quad |\itar \cap V(G_j)| \ge \left(1 - \frac{2}{k+1}\right)\eta.
\]

We construct an $(\iini,\itar)$-reconfiguration sequence $\mathscr{S}$ for $G$ as follows:
\begin{description}
    \item[Step 1.] Remove all vertices of $\iini \cap V_j$ from $\iini$ to obtain $\iini \cap V(G_j)$.
    \item[Step 2.] Compute a tree decomposition of $G_j$ and use the algorithm in \Cref{thm:algo:treewidth} to compute an $(\iini \cap V(G_j), \itar \cap V(G_j))$-reconfiguration sequence $\mathscr{S}_j$ on $G_j$. 
    \item[Step 3.] Add all vertices of $\itar \cap V_j$ to obtain $\itar$.
\end{description}

We now analyse the approximation factor.
Let $t \coloneq \tw(G_j) + 1$ and $\varphi \coloneq \min\{|\iini\cap V(G_j)|, |\itar\cap V(G_j)|\}$.
Recall that the choice of $j$ ensures that
\begin{align}
    \label{ineq:intersection_bound}
    \varphi \geq \left(1-\frac{2}{k+1}\right)\eta.
\end{align}

Define the function 
$g(x) \coloneq \left(1 - \frac{1}{1+\varepsilon'}\right)x - t\left(\log_{3/2}\frac{x}{t} + 1\right)$, where $\varepsilon'$ is a constant satisfying $0<\varepsilon'<\varepsilon$.
For example, we can set $\varepsilon' = \varepsilon/2$.
By \Cref{thm:algo:treewidth} and \Cref{ineq:intersection_bound}, we obtain 
\begin{align}
    \label{ineq:val_lower_bound_1}
    \val_G(\mathscr{S}) = \val_{G_j}(\mathscr{S}_j) 
    \geq g(\varphi) + \frac{1}{1+\varepsilon'}\varphi \ge g(\varphi) + \frac{1}{1+\varepsilon'} \cdot \left(1-\frac{2}{k+1}\right) \eta.
\end{align}

We now choose $k$ such that $\frac{1}{1+\varepsilon'} \cdot \left(1-\frac{2}{k+1}\right) \geq \frac{1}{1+\epsilon}$.
For example, it suffices to set $k = 2\lceil \frac{1+\epsilon}{\epsilon-\epsilon'} \rceil-1$.
With this choice of $k$, we conclude that
\begin{align}
    \label{ineq:val_lower_bound}
    \val_G(\mathscr{S})  
    \geq  g(\varphi) + \frac{1}{1+\varepsilon'} \cdot \left(1-\frac{2}{k+1}\right) \eta \geq g(\varphi) + \frac{1}{1+\epsilon} \eta.
\end{align}

We now distinguish two cases.

\medskip
\noindent
\textbf{Case~1: $g(\varphi) \geq 0$.}
In this case, \Cref{ineq:val_lower_bound} directly implies that 
\[
    \val_G(\mathscr{S}) \geq \frac{1}{1+\varepsilon}\eta \geq \frac{1}{1+\varepsilon}\opt_G(\iini \reco \itar).
\]
Thus, the algorithm achieves the desired approximation ratio.

\medskip
\noindent
\textbf{Case~2: $g(\varphi) < 0$.}
A straightforward calculation yields the second derivative function $g''(x) = \frac{t}{x^2\ln(3/2)}$.
Since $g''(x) > 0$ for all $x \in \mathbb{R}_{>0}$ assuming $t > 0$, it follows that $g$ is convex on $\mathbb{R}_{>0}$. 
Moreover, since $1 - \frac{1}{1+\epsilon'} = \frac{\epsilon'}{1+\epsilon'}> 0$,
we have $\lim_{x \to +\infty} g(x) = +\infty$.
Together with the assumption $g(\varphi)<0$, this implies that $\varphi$ is bounded by a function depending only on $t$ and $\varepsilon'$.

Recall that $t=\tw(G_j)+1 \le c_h \cdot k$, where $k = 2\left\lceil \frac{1+\varepsilon}{\varepsilon-\varepsilon'} \right\rceil - 1$ and $\varepsilon' < \varepsilon$.
Consequently, $\varphi$ 
is bounded by a function depending only on $\varepsilon$ and $h$.
It is known that every $H$-minor-free graph has degeneracy $O(h\log h)$~\cite{Thomason2001}.
Since $\varphi$
is bounded, we can apply \Cref{thm:FPT_degeneracy} to solve the instance exactly.

\medskip
In both cases, we obtain an $(\iini,\itar)$-reconfiguration sequence
whose objective value is at least $\opt_G(\iini\reco\itar)/(1+\varepsilon)$.

We conclude by analysing the running time of the algorithm.
By \Cref{lem:partition_minor_free}, the partition $V_1, V_2, \ldots, V_{k+1}$ can be computed in polynomial time.
The selection of the index $j$, as well as Steps~1 and~3, also requires only polynomial time.
In Step~2, we have $\tw(G_j) \le c_h \cdot k$.
Thus, a tree decomposition of $G_j$ can be computed in \fpt time parameterised by $\varepsilon$ and $h$~\cite{Bodlaender96}.
Given such a decomposition, the algorithm of \Cref{thm:algo:treewidth} runs in polynomial time.
In Case~2, we additionally invoke the FPT algorithm of \Cref{thm:FPT_degeneracy}, parameterised by $\varepsilon$ and $h$.
Therefore, the overall running time of the algorithm is \fpt time with respect to $\varepsilon$ and $h = |V(H)|$.
This completes the proof of \Cref{thm:minor_free}.
\end{proof}

\begin{remark}
    Consider the special case of planar graphs.
    On the one hand, since every $n$-vertex planar graph has at most $3n-6$ edges, it is 5-degenerate.
    Hence, \cref{thm:algo:degeneracy} implies a polynomial-time algorithm that, given a planar graph $G$ and two independent sets $\iini$ and $\itar$ of $G$, produces an $(\iini, \itar)$-reconfiguration sequence $\mathscr{J}$ such that
    \[
        \val_G(\mathscr{J}) \geq \frac{1}{5} \cdot \opt_G(\iini \reco \itar) - 1.
    \]
    On the other hand, since every planar graph is $K_5$-minor-free, if we allow a slower run time (i.e., \fpt time parameterised by $\epsilon$ instead of polynomial time), we can use \cref{thm:minor_free} to obtain a $(1+\epsilon)$-approximation algorithm for planar graphs for any $\epsilon > 0$.
\end{remark}

\section{Inapproximability results}
\label{sec:hardness}

In this section,
we present inapproximability results for \mmisr on
bounded-degree graphs,
graphs of bandwidth $n^{\frac{1}{2}+\Theta(1)}$, and
bipartite graphs.

\subsection{Bounded degree graphs}
\label{sec:hardness:degree}

We show that for a graph of maximum degree $\Delta$,
\mmisr is $\NP$-hard to approximate within a factor of $\Theta\bigl(\sqrt{\Delta}\bigr)$.

\begin{theorem}
\label{thm:hardness:degree}
For a graph of maximum degree $\Delta \geq 3$,
it is $\NP$-hard (under randomised reductions)
to approximate \mmisr  
within a factor of $\Theta\bigl(\sqrt{\Delta}\bigr)$.
\end{theorem}

The proof of \cref{thm:hardness:degree} is based on a gap-preserving reduction
from \mis on bounded-degree graphs,
which exhibits the following $\NP$-hardness of approximation
due to Bhangale and Khot \cite{bhangale2022ug}.

\begin{theorem}[\cite{bhangale2022ug}]
\label{thm:BK22}
For an $n$-vertex graph $G$ of maximum degree $\Delta$,
it is $\NP$-hard (under randomised reductions) to distinguish whether
$\alpha(G) \geq \Theta\left(\tfrac{1}{\log \Delta}n\right)$
or
$\alpha(G) \leq \Theta\left(\tfrac{\log \Delta}{\Delta}n\right)$.
\end{theorem}

\subparagraph{Expander graphs.}

Before describing the gap-preserving reduction, we introduce \emph{bipartite expander graphs}.
Let $G$ be an $n$-vertex $d$-regular (multi)graph.
For each $i \in [n]$, let $\lambda_i(G)$ be the $i$-th largest eigenvalue of its adjacency matrix.
Note that $\lambda_1(G) = d$ and
that $\lambda_n(G) = -d$ if and only if $G$ is bipartite.
See also the survey of Hoory, Linial, and Wigderson \cite{hoory2006expander}.

\begin{definition2}[Bipartite expander graph]
For positive integers $n$ and $d \geq 3$ and 
a positive real $\lambda < d$,
a \defi{bipartite $(n,d,\lambda)$-expander graph}
is defined as a $2n$-vertex $d$-regular balanced bipartite (multi)graph $G$ such that
$|\lambda_i(G)| \leq \lambda$ for every $2 \leq i \leq 2n-1$.\footnote{
Note that a bipartite $(n,d,\lambda)$-expander graph contains $2n$ vertices.
}
\lipicsEnd
\end{definition2}

A bipartite $(n,d,\lambda)$-expander graph is called \defi{Ramanujan} \cite{lubotzky1988ramanujan}
if $\lambda \leq 2\sqrt{d-1}$; namely,
$|\lambda_i(G)| \leq 2\sqrt{d-1}$ for every $2 \leq i \leq 2n-1$.
Cohen \cite{cohen2016ramanujan},
along with the interlacing polynomials method due to {Marcus, Spielman, and Srivastava}
\cite{marcus2015interlacing,marcus2018interlacing},
demonstrated an explicit family of $d$-regular bipartite Ramanujan (multi)graphs
for every $d \geq 3$.

\begin{theorem}[\cite{cohen2016ramanujan,marcus2015interlacing,marcus2018interlacing}]
\label{thm:Coh16}
For any positive integers $d \geq 3$ and $n$,
one can construct a $2n$-vertex $d$-regular bipartite Ramanujan graph in polynomial time in $n$.
\end{theorem}

For a bipartite graph $G$ with bipartition $(L,R)$ and
any two vertex sets $S \subseteq L$ and $T \subseteq R$,
let $e_G(S,T)$ denote the number of edges between $S$ and $T$; namely,
\begin{align}
    e_G(S,T)
    \defeq \bigl|(S \times T) \cap E(G)\bigr|
    = \left|\bigl\{ (v,w) \in S \times T \bigm| (v,w) \in E(G) \bigr\}\right|.
\end{align}
The bipartite expander mixing lemma \cite{haemers1979eigenvalue,haemers1995interlacing} states that
$e_G(S,T)$ of an expander graph $G$ is close to the expected value in the corresponding random $d$-regular bipartite graph.

\begin{lemma}[Bipartite expander mixing lemma \cite{haemers1979eigenvalue,haemers1995interlacing}]
\label{lem:EML}
Let $G$ be a bipartite $(n,d,\lambda)$-expander graph with bipartition $(L,R)$.
For any two vertex sets $S \subseteq L$ and $T \subseteq R$,
it holds that
\begin{align}
    \left|e_G(S,T) - \frac{d \cdot |S|\cdot|T|}{n}\right| \leq \lambda \sqrt{|S|\cdot|T|}.
\end{align}
\end{lemma}

\subparagraph{Reduction.}
Our gap-preserving reduction from \mis to \mmisr is described as follows.
Let $G$ be an $n$-vertex graph of maximum degree $\Delta \geq 3$.
Define $m \defeq \Theta\left(\tfrac{\log \Delta}{\sqrt{\Delta}}n\right)$.
Let $X$ be a $2m$-vertex $\Delta$-regular bipartite Ramanujan graph, i.e.,
a bipartite $(m,\Delta,\lambda)$-expander graph with bipartition $(L,R)$,
where $\lambda \defeq 2\sqrt{\Delta-1}$ and $|L|=|R|=m$.
By \cref{thm:Coh16}, such $X$ can be constructed in $\poly(m)$ time.
Let $H \defeq G + X$ be the disjoint union of $G$ and $X$.
The number of vertices in $H$ is $|V(H)| = n+2m = \Theta(n)$, and
the maximum degree of $H$ is at most $\Delta$.
The initial and target independent sets are defined as $\iini \defeq L$ and $\itar \defeq R$.
This completes the description of the reduction.

\subparagraph{Correctness.}
We first show the following completeness.
\begin{lemma}
\label{lem:hardness:degree:completeness}
If $\displaystyle\alpha(G) \geq \Theta\left(\tfrac{1}{\log \Delta}n\right)$,
then $\displaystyle\opt_H(\iini \reco \itar) \geq \Theta\left(\tfrac{\log \Delta}{\sqrt{\Delta}}n\right)$.
\end{lemma}
\begin{proof}
Let $I_G$ be a maximum independent set of $G$, whose size is $\alpha(G)$.
Consider an $(\iini,\itar)$-reconfiguration sequence $\sqI$ obtained by the following procedure:
\begin{description}
    \item[Step 1.] add all $\alpha(G)$ vertices of $I_G$ 
    to obtain $L \cup I_G$.
    \item[Step 2.] remove all $m$ vertices of $L$ 
    to obtain $I_G$.
    \item[Step 3.] add all $m$ vertices of $R$ 
    to obtain $R \cup I_G$.
    \item[Step 4.] remove all $\alpha(G)$ vertices of $I_G$ 
    to obtain $R$.
\end{description}
The value of $\sqI$ can be bounded from below as
\begin{align}
\begin{aligned}
    \val_H(\sqI)
    = \min\bigl\{ |L|, |R|, |I_G| \bigr\}
    = \min\left\{ \Theta\left(\tfrac{\log \Delta}{\sqrt{\Delta}}n\right), \Theta\left(\tfrac{1}{\log \Delta}n\right) \right\}
    = \Theta\left(\tfrac{\log \Delta}{\sqrt{\Delta}}n\right),
\end{aligned}
\end{align}
as desired.
\end{proof}

We then show the following soundness.
\begin{lemma}
\label{lem:hardness:degree:soundness}
    If $\displaystyle\alpha(G) \leq \Theta\left(\tfrac{\log \Delta}{\Delta}n\right)$,
    then $\displaystyle\opt_H(\iini \reco \itar) \leq \Theta\left(\tfrac{\log \Delta}{\Delta}n\right)$.
\end{lemma}
\begin{proof}
Define
\begin{align}
    \epsilon_\Delta
    \defeq \frac{2.01 \sqrt{\Delta-1}}{\Delta}
    = \Theta\left(\frac{1}{\sqrt{\Delta}}\right).
\end{align}
Note that $\epsilon_\Delta \in (0,1)$ for $\Delta \geq 3$.
We use the following claim.
\begin{claim}
\label{clm:hardness:degree}
    Let $S \subseteq L$ and $T \subseteq R$ be any two vertex sets such that 
    $|S| \geq \epsilon_\Delta m$ and $|T| \geq \epsilon_\Delta m$.
    Then, $S \cup T$ is not an independent set of $X$.
\end{claim}
\begin{claimproof}
It is sufficient to consider the case where $|S|=|T|$.
Let $\epsilon \defeq \frac{|S|}{m} = \frac{|T|}{m}$.
By assumption, $\epsilon \geq \epsilon_\Delta$.
By \cref{lem:EML}, we have
\begin{align}
\begin{aligned}
    e_X(S,T)
    & \geq \frac{\Delta \cdot |S|\cdot |T|}{m} - \lambda \sqrt{|S| \cdot |T|} \\
    & = \frac{\Delta \cdot \epsilon m \cdot \epsilon m}{m} - \lambda\sqrt{\epsilon m \cdot \epsilon m} \\
    & \geq (\Delta \epsilon_\Delta - \lambda)\epsilon_\Delta m \\
    & = 0.01 \sqrt{\Delta-1} \cdot \epsilon_\Delta m > 0,
\end{aligned}
\end{align}
implying that $X[S \cup T]$ contains an edge; i.e.,
$S \cup T$ is not an independent set of $X$, as desired.
\end{claimproof}
Consider now any optimal $(\iini,\itar)$-reconfiguration sequence $\sqI^* = (I^{(1)}, \ldots, I^{(t)})$.
Without loss of generality, we assume that every two consecutive sets differ by exactly one vertex.
Since $|I^{(1)} \cap L| = m$ and $|I^{(t)} \cap L| = 0$,
there exists $i \in [t]$ such that
$|I^{(i)} \cap L| = \lceil \epsilon_\Delta m \rceil$.
By \cref{clm:hardness:degree}, $|I^{(i)} \cap R| \leq \lceil \epsilon_\Delta m \rceil$.
The size of $I^{(i)}$ is bounded as
\begin{align}
    \bigl|I^{(i)}\bigr|
    = \underbrace{\bigl|I^{(i)} \cap L\bigr|}_{= \lceil \epsilon_\Delta m \rceil}
    + \underbrace{\bigl|I^{(i)} \cap R\bigr|}_{\leq \lceil \epsilon_\Delta m \rceil}
    + \underbrace{\bigl|I^{(i)} \cap V(G)\bigr|}_{\leq \alpha(G)}
    \leq 2 \lceil \epsilon_\Delta m \rceil + \alpha(G).
\end{align}
Consequently, we have
\begin{align}
\begin{aligned}
    \val_H(\sqI^*)
    & \leq \bigl|I^{(i)}\bigr|
    \leq 2 \lceil \epsilon_\Delta m \rceil + \alpha(G) \\
    & \underbrace{\leq}_{\epsilon_\Delta = \Theta\left(\frac{1}{\sqrt{\Delta}}\right)}
            \Theta\left(\tfrac{1}{\sqrt{\Delta}}\tfrac{\log \Delta}{\sqrt{\Delta}}n\right) + \Theta\left(\tfrac{\log \Delta}{\Delta} n\right) \\
    & = \Theta\left(\tfrac{\log \Delta}{\Delta} n\right),
\end{aligned}
\end{align}
as desired.
\end{proof}

We are now ready to prove \cref{thm:hardness:degree}.

\begin{proof}[Proof of \cref{thm:hardness:degree}]
By \cref{thm:BK22},
for an $n$-vertex graph $G$ of maximum degree $\Delta \geq 3$,
it is $\NP$-hard (under randomised reductions) to distinguish whether
$\alpha(G) \geq \Theta\left(\tfrac{1}{\log \Delta}n\right)$
or
$\alpha(G) \leq \Theta\left(\tfrac{\log \Delta}{\Delta}n\right)$.
Let $(H, \iini, \itar)$ be an instance of \mmisr obtained by applying the reduction above to $G$.
Note that the maximum degree of $H$ is $\Delta$.
By \cref{lem:hardness:degree:completeness,lem:hardness:degree:soundness},
it is $\NP$-hard to distinguish whether
$\opt_G(\iini \reco \itar) \geq \Theta\left(\tfrac{\log \Delta}{\sqrt{\Delta}}n\right)$
or
$\opt_G(\iini \reco \itar) \leq \Theta\left(\tfrac{\log \Delta}{\Delta}n\right)$; i.e.,
it is $\NP$-hard to approximate \mmisr with in a factor of
$\Theta\bigl(\sqrt{\Delta}\bigr)$,
which completes the proof.
\end{proof}

\subsection{Bandwidth $n^{\frac{1}{2}+\Theta(1)}$}
\label{sec:hardness:bandwidth}

We show that
for an $n$-vertex graph of bandwidth $n^{\frac{1}{2}+\Theta(1)}$,
\mmisr is hard to approximate within a polynomial factor.

\begin{theorem}
\label{thm:hardness:bandwidth}
Let $\delta \in \left(0,\frac{1}{2}\right)$ be any positive real and
$\epsilon \in \left(0, \frac{1}{2}-\delta\right)$ be any small positive real.
For an $n$-vertex graph of bandwidth $\bigO\bigl(n^{\frac{1}{2}+\delta}\bigr)$,
it is $\NP$-hard to approximate \mmisr within a factor of
$\Theta\bigl(n^{4\delta^2-\epsilon}\bigr)$.
\end{theorem}

To prove \cref{thm:hardness:bandwidth},
we first show the following $\NP$-hardness of \mis.

\begin{lemma}
\label{lem:hardness:bandwidth}
    Let $\delta \in \left(0,\frac{1}{2}\right)$ be any positive real and
    $\epsilon \in (0, 1-2\delta)$ be any small positive real.
    For an $n$-vertex graph $G$ of bandwidth $\bigO\bigl(n^{\frac{1}{2}+\delta}\bigr)$,
    it is $\NP$-hard to distinguish whether
    $\alpha(G) \geq \Theta\bigl(n^{\frac{1-\epsilon}{2-2\delta}}\bigr)$ or
    $\alpha(G) \leq \Theta\bigl(n^{\frac{1-2\delta}{2-2\delta}}\bigr)$.
\end{lemma}
\begin{proof}
We describe a gap-preserving reduction from \mis to itself.
Let $\delta \in \left(0, \frac{1}{2}\right)$ and $\epsilon \in (0, 1-2\delta)$ be two positive reals.
Let $G$ be an $n$-vertex graph,
$C$ be a graph consisting of $\bigl\lceil n^{1-2\delta} \bigr\rceil$ $n$-cliques $K_n$, and
$H \defeq G + C$ be the disjoint union of $G$ and $C$.
The number of vertices in $H$, denoted by $m \defeq |V(H)|$, is 
\begin{align}
    m
    = |V(G)| + |V(C)|
    = n + n\cdot \bigl\lceil n^{1-2\delta} \bigr\rceil
    = \Theta\bigl(n^{2-2\delta}\bigr).
\end{align}
The bandwidth of $H$ is 
\begin{align}
    \bw(H)
    = \max\bigl\{ \bw(G), \bw(C) \bigr\}
    = \max\bigl\{ n, \bw(K_n) \bigr\}
    = n
    = \Theta\bigl(m^{\frac{1}{2-2\delta}}\bigr)
    = \bigO\bigl(m^{\frac{1}{2}+\delta}\bigr),
\end{align}
where we used the inequality that $\frac{1}{2-2\delta} < \frac{1}{2}+\delta$ for $\delta \in \left(0,\frac{1}{2}\right)$.
This completes the description of the reduction.

On one hand, if $\alpha(G) \geq n^{1-\epsilon}$, then
\begin{align}
    \alpha(H)
    \geq \alpha(G)
    \geq n^{1-\epsilon}
    = \Theta\bigl(m^{\frac{1-\epsilon}{2-2\delta}}\bigr).
\end{align}
\item On the other hand, if $\alpha(G) \leq n^{\epsilon}$, then
\begin{align}
\begin{aligned}
    \alpha(H)
    = \alpha(G) + \alpha(C)
    \leq n^{\epsilon} + \bigl\lceil n^{1-2\delta} \bigr\rceil
    \underbrace{\leq}_{\epsilon < 1-2\delta} \Theta\bigl(n^{1-2\delta}\bigr)
    = \Theta\bigl(m^{\frac{1-2\delta}{2-2\delta}}\bigr).
\end{aligned}
\end{align}
By \cite{MR1687331,zuckerman2007linear}, for an $n$-vertex graph $G$,
it is $\NP$-hard to distinguish whether
$\alpha(G) \geq n^{1-\epsilon}$ or $\alpha(G) \leq n^{\epsilon}$
for any small positive real $\epsilon$.
Therefore,
for an $m$-vertex graph $H$ of bandwidth $\bigO\bigl(m^{\frac{1}{2}+\delta}\bigr)$,
it is $\NP$-hard to distinguish whether
$\alpha(H) \geq \Theta\bigl(m^{\frac{1-\epsilon}{2-2\delta}}\bigr)$ or
$\alpha(H) \leq \Theta\bigl(m^{\frac{1-2\delta}{2-2\delta}}\bigr)$,
completing the proof.
\end{proof}

We are now ready to prove \cref{thm:hardness:bandwidth}.

\begin{proof}[Proof of \cref{thm:hardness:bandwidth}]
We describe a gap-preserving reduction from \mis to \mmisr.
Let $\delta \in \left(0, \frac{1}{2}\right)$ be a positive real and
$G$ be an $n$-vertex graph of bandwidth $\bw(G) = \bigO\bigl(n^{\frac{1}{2}+\delta}\bigr)$.
Define $k \defeq \Theta\bigl( n^{\frac{1-\epsilon}{2-2\delta}} \bigr)$,
where $\epsilon$ is a (sufficiently small) positive real, whose value will be determined later.
Let $H$ be the disjoint union of $G$ and a complete bipartite graph $K_{k,k}$ with bipartition $(L,R)$.
The number of vertices in $H$, denoted by $m \defeq |V(H)|$, is
\begin{align}
    m = |V(G)| + |V(K_{k,k})| = n + 2k = \Theta(n).
\end{align}
The bandwidth of $H$ is bounded as
\begin{align}
    \bw(H)
    = \max\bigl\{ \bw(G), \bw(K_{k,k}) \bigr\}
    = \max\Bigl\{ \bigO\bigl(n^{\frac{1}{2}+\delta}\bigr), \Theta\bigl( n^{\frac{1-\epsilon}{2-2\delta}} \bigr) \Bigr\}
    = \bigO\bigl(n^{\frac{1}{2}+\delta}\bigr)
    = \bigO\bigl(m^{\frac{1}{2}+\delta}\bigr).
\end{align}
The initial and target independent sets are defined as $\iini \defeq L$ and $\itar \defeq R$.
We obtain an instance $(H,\iini,\itar)$ of \mmisr, completing the description of the reduction.

We first show the completeness.
Suppose that $\alpha(G) \geq \Theta\bigl(n^{\frac{1-\epsilon}{2-2\delta}}\bigr)$.
Let $I_G$ be a maximum independent set of $G$, whose size is $\alpha(G)$.
Consider an $(\iini,\itar)$-reconfiguration sequence $\sqI$ obtained by the following procedure:
\begin{description}
    \item[Step 1.] add all $\alpha(G)$ vertices of $I_G$
    to obtain $L \cup I_G$.
    \item[Step 2.] remove all $k$ vertices of $L$
    to obtain $I_G$.
    \item[Step 3.] add all $k$ vertices of $R$ 
    to obtain $R \cup I_G$.
    \item[Step 4.] remove all $\alpha(G)$ vertices of $I_G$
    to obtain $R$.
\end{description}
The value of $\sqI$ can be bounded from below as
\begin{align}
    \val_H(\sqI)
    = \min\bigl\{|L|, |R|, |I_G|\bigr\}
    \geq \Theta\bigl(n^{\frac{1-\epsilon}{2-2\delta}}\bigr)
    = \Theta\bigl(m^{\frac{1-\epsilon}{2-2\delta}}\bigr).
\end{align}
Therefore, $\opt_H(\iini \reco \itar) \geq \Theta\bigl(m^{\frac{1-\epsilon}{2-2\delta}}\bigr)$.

We then show the soundness.
Suppose that $\alpha(G) \leq \Theta\bigl(n^{\frac{1-2\delta}{2-2\delta}}\bigr)$.
Consider any optimal $(\iini, \itar)$-reconfiguration sequence $\sqI^* = (I^{(1)}, \ldots, I^{(t)})$.
Without loss of generality, we assume that every two consecutive sets differ by exactly one vertex.
Since $L \cup R$ induces a complete bipartite graph $K_{k,k}$,
there must exist $i \in [t]$ such that $|I^{(i)} \cap L| = |I^{(i)} \cap R| = 0$; i.e.,
$|I^{(i)}| = |I^{(i)} \cap V(G)| \leq \alpha(G)$.
The value of $\sqI$ can be bounded from above as 
\begin{align}
    \val_H(\sqI^*)
    \leq \bigl|I^{(i)}\bigr|
    \leq \alpha(G)
    \leq \Theta\bigl(n^{\frac{1-2\delta}{2-2\delta}}\bigr)
    = \Theta\bigl(m^{\frac{1-2\delta}{2-2\delta}}\bigr).
\end{align}
Therefore, $\opt_H(\iini \reco \itar) \leq \Theta\bigl(m^{\frac{1-2\delta}{2-2\delta}}\bigr)$.

By \cref{lem:hardness:bandwidth},
for an $m$-vertex graph $H$ of bandwidth $\bigO\bigl(m^{\frac{1}{2}+\delta}\bigr)$ and
a pair of its independent sets $\iini$ and $\itar$,
it is $\NP$-hard to distinguish whether
$\opt_H(\iini \reco \itar) \geq \Theta\bigl(m^{\frac{1-\epsilon}{2-2\delta}}\bigr)$
or
$\opt_H(\iini \reco \itar) \leq \Theta\bigl(m^{\frac{1-2\delta}{2-2\delta}}\bigr)$.
Therefore, it is $\NP$-hard to approximate \mmisr within a factor of 
\begin{align}
\label{eq:hardness:bandwidth:factor}
    \frac{\Theta\bigl(m^{\frac{1-\epsilon}{2-2\delta}}\bigr)}{\Theta\bigl(m^{\frac{1-2\delta}{2-2\delta}}\bigr)}
    = \Theta\bigl(m^{\frac{2\delta-\epsilon}{2-2\delta}}\bigr).
\end{align}
Letting $\epsilon \defeq \epsilon' (2-2\delta) > 0$
for any small positive $\epsilon' < \frac{1}{2}-\delta < \frac{1-2\delta}{2-2\delta}$
so that $\epsilon < 1-2\delta$,
we obtain
\begin{align}
    \text{\cref{eq:hardness:bandwidth:factor}}
    = \Theta\bigl(m^{\frac{2\delta-\epsilon}{2-2\delta}}\bigr)
    = \Theta\bigl(m^{\frac{\delta}{1-\delta}-\epsilon'}\bigr)
    \geq \Omega\bigl(m^{4\delta^2-\epsilon'}\bigr),
\end{align}
where we used the inequality that $\frac{\delta}{1-\delta} > 4\delta^2$ for $\delta \in \left(0,\frac{1}{2}\right)$,
which completes the proof.
\end{proof}

\subsection{Bipartite graphs}
\label{sec:hardness:bipartite}

We show that for an $n$-vertex bipartite graph,
\mmisr cannot be approximated within a factor of $n^{1-\epsilon}$ in polynomial time,
under the Small Set Expansion Hypothesis (SSEH) \cite{raghavendra2010graph} and $\NP \nsubseteq \BPP$.

\begin{theorem}
\label{thm:hardness:bipartite}
    Assuming SSEH and $\NP \nsubseteq \BPP$,
    for an $n$-vertex bipartite graph,
    no polynomial-time algorithm can approximate \mmisr within a factor of $n^{1-\epsilon}$ for any $\epsilon > 0$.
\end{theorem}

The proof of \cref{thm:hardness:bipartite} is based on
a gap-preserving reduction from \textsc{Maximum Balanced Biclique}, which we define as follows.
We say that a bipartite graph with bipartition $(L,R)$ is \defi{balanced} if $|L| = |R|$.
For a bipartite graph $G$, the \textsc{Maximum Balanced Biclique} (\mbb) problem asks to find
a balanced biclique (i.e., a balanced complete bipartite graph $K_{k,k}$) of $G$ with the maximum number of vertices.
Manurangsi~\cite{Manurangsi18biclique} proved the following inapproximability result.

\begin{lemma}[\protect{\cite[Lemma A1]{Manurangsi18biclique}}]
\label{thm:Man18}
    Assuming SSEH and $\NP \nsubseteq \BPP$,
    for any $\epsilon > 0$ and for a balanced bipartite graph with bipartition $(L,R)$,
    no polynomial-time algorithm can distinguish whether
    $G$ contains a balanced biclique $K_{n^{1-\epsilon}, n^{1-\epsilon}}$ as a subgraph or
    $G$ does not contain a balanced biclique $K_{n^\epsilon, n^\epsilon}$ as a subgraph,
    where $n=|L|=|R|$.
\end{lemma}

\begin{proof}[Proof of \cref{thm:hardness:bipartite}]
We describe a gap-preserving reduction from \mbb to \mmisr on bipartite graphs.
Let $G = (L, R, E)$ be a balanced bipartite graph as an instance of \mbb.
Define $n \defeq |L| = |R|$ and $m \defeq |V(G)| = 2n$.
An instance of \mmisr is defined as $(H, \iini, \itar)$, where
$H \defeq (L, R, (L \times R) \setminus E)$ is a balanced bipartite graph and
$\iini \defeq L$ and $\itar \defeq R$ are the initial and target independent sets.

We first show the completeness.
Suppose that $G$ contains a balanced biclique $K_{n^{1-\epsilon},n^{1-\epsilon}}$ as a subgraph,
denoted by $S \cup T$ such that $S \subseteq L$, $T \subseteq R$, and $|S|=|T| \geq n^{1-\epsilon}$.
Note that $S \cup T$ is an independent set of $H$.
Consider an $(\iini,\itar)$-reconfiguration sequence $\sqI$
passing through $S \cup T$ obtained by the following procedure:
\begin{description}
    \item[Step 1.] remove all vertices of $L \setminus S$
    to obtain $S$.
    \item[Step 2.] add all vertices of $T$
    to obtain $S \cup T$.
    \item[Step 3.] remove all vertices of $S$
    to obtain $T$.
    \item[Step 4.] add all vertices of $R \setminus T$
    to obtain $R$.
\end{description}
The value of $\sqI$ is bounded from below as
\begin{align}
    \val_H(\sqI)
    = \min\bigl\{|L|, |R|, |S|, |T| \bigr\}
    \geq n^{1-\epsilon}.
\end{align}
Therefore, $\opt_H(\iini \reco \itar) \geq n^{1-\epsilon}$.

We then show the soundness.
Suppose that $G$ does not contain a balanced biclique $K_{n^\epsilon,n^\epsilon}$ as a subgraph.
Consider any optimal $(\iini,\itar)$-reconfiguration sequence $\sqI^* = (I^{(1)}, \ldots, I^{(t)})$.
Without loss of generality, we assume that every two consecutive sets differ by exactly one vertex.
Since $|I^{(1)} \cap L| = n$ and $|I^{(t)} \cap L| = 0$,
there must exist $i \in [t]$ such that $|I^{(i)} \cap L| = \bigl\lceil n^\epsilon \bigr\rceil$.
Since $G$ does not contain $K_{n^\epsilon,n^\epsilon}$ as a subgraph, we have $|I^{(i)} \cap R| < \bigl\lceil n^\epsilon \bigr\rceil$, implying that
\begin{align}
    \val_H(\sqI^*)
    \leq \bigl|I^{(i)}\bigr|
    < 2 \bigl\lceil n^\epsilon \bigr\rceil.
\end{align}
Therefore, $\opt_H(\iini \reco \itar) < 2 \bigl\lceil n^\epsilon \bigr\rceil$.

By applying \cref{thm:Man18},
under SSEH and $\NP \nsubseteq \BPP$,
no polynomial-time algorithm can distinguish whether
$\opt_H(\iini \reco \itar) \geq n^{1-\epsilon} \geq m^{1-2\epsilon}$ or
$\opt_H(\iini \reco \itar) < 2 \bigl\lceil n^\epsilon \bigr\rceil \leq m^{2\epsilon}$
for sufficiently large $m = 2n$.
Therefore, under the same assumption,
no polynomial-time algorithm can approximate \mmisr on an $m$-vertex bipartite graph
within a factor of $m^{1-4\epsilon}$ for any $\epsilon > 0$,
which completes the proof.
\end{proof}

\bibliography{refs}

\end{document}